\documentclass[sigconf,nonacm=true,natbib=false]{acmart}

\usepackage{algorithm}
\usepackage{algpseudocode}
\usepackage{subcaption}
\usepackage{stfloats}
\usepackage[normalem]{ulem}   
\usepackage{xcolor}           
\usepackage{colortbl}         
\usepackage{booktabs}         
\usepackage{makecell}         

\AtBeginDocument{%
  }

\setcopyright{acmlicensed}
\copyrightyear{2026}
\acmYear{2026}
\acmDOI{XXXXXXX.XXXXXXX}
\acmConference[CCS'26]{}{November 2026}{Den Haag, The Netherlands}
\acmISBN{978-1-4503-XXXX-X/2018/06}

\RequirePackage[
  datamodel=acmdatamodel,
  style=acmnumeric,
  ]{biblatex}

\begin{document}

\title{Herring: Parallel Batch-Order-Fairness on DAG-based Blockchain Consensus}


\author{Marko Putnik}
\affiliation[obeypunctuation=true]{%
  \institution{Delft University of Technology},
  \country{The Netherlands}%
}
\email{M.Putnik@student.tudelft.nl}

\author{Jérémie Decouchant}
\affiliation[obeypunctuation=true]{%
  \institution{Delft University of Technology},
  \country{The Netherlands}%
}
\email{J.Decouchant@tudelft.nl}








\begin{abstract}
Transaction ordering attacks extract billions of dollars annually from decentralized finance users in the form of Maximal Extractable Value (MEV). Byzantine Fault-Tolerant (BFT) consensus protocols guarantee total order but place no constraint on how that order is chosen, leaving the door open for adversarial reordering. Batch-order-fairness (batch-OF) protocols close this gap, but existing designs pay a steep performance price for this guarantee. Leader-based protocols such as Themis concentrate all fairness decisions at a single replica, while recent DAG-based proposals FairDAG and DAG of DAGs (DoD) force their fairness layer into strictly serial execution despite running on multi-proposer DAGs.

We present Herring, the first $\gamma$-batch-OF DAG BFT protocol whose fairness layer parallelizes the dominant graph construction cost across committed subdags. Herring combines post-consensus graph construction with explicit missing edge resolution piggybacked on the DAG's reliable broadcast layer, a pairing that turns fair ordering from a per-round serial bottleneck into a CPU-bound task. We also uncover previously unreported liveness vulnerabilities in both FairDAG-RL and DoD that a malicious client can trigger to halt the fairness layer indefinitely, and propose patches that we integrate into our reimplementations.

We implement Herring on top of the Rust implementation of Narwhal \& Tusk and evaluate it against FairDAG-RL, DoD-W, and Themis. Herring tracks the throughput of Narwhal \& Tusk closely up to roughly $10{,}000$\,tx/s, achieves roughly $90\%$ higher saturation throughput than FairDAG-RL and $100\%$ higher than DoD-W, and substantially reduces execution latency at saturation. Our adversarial evaluation shows that Herring confines adversarial reordering to only the most fragile transaction pairs, a robustness that stems from the combination of its batch-OF guarantees and the underlying DAG's dissemination layer.
\end{abstract}





\maketitle


\newcommand{\good}[1]{\textcolor{green!55!black}{\textbf{#1}}}
\newcommand{\bad}[1]{\textcolor{red!60!black}{\uwave{#1}}}

%

%

\section{Introduction}

\pagenumbering{arabic}

Decentralized Finance (DeFi) has grown into a multi-billion-dollar ecosystem built on top of permissionless blockchains, where the order in which transactions are committed often determines who profits and who pays.
Transaction ordering attacks such as front-running, back-running, and sandwich attacks have stolen billions of dollars from ordinary users in the form of Maximal Extractable Value (MEV)~\cite{MEV_theoretical, chainlinkMEV2023}. Byzantine Fault-Tolerant (BFT) consensus protocols~\cite{pbft, hotstuff, Narwhal&Tusk_Mar2022} guarantee that all honest replicas agree on a single total order of transactions, but place no constraint on how that order is chosen relative to the order in which transactions arrived at the network.
This gap has been extensively exploited.
In leader-based protocols such as HotStuff~\cite{hotstuff}, a leader can unilaterally reorder or delay transactions within its block without violating any consensus property.

To close this gap, a recent line of work~\cite{Kelkar_Cornell_Aequitas_Aug2020,Kelkar_Cornell_Themis_Nov2022,Zhang_Cornell_Pompe_Nov2020,Cachin_Bern_2024,Kiayias_Edinburgh_Taxis_2024,Chen_China_Auncel_Aug2024} has introduced the notion of \emph{order-fairness (OF)} as an additional safety property for BFT consensus, proposing several definitions that differ in strength and tractability.
Among these, $\gamma$-batch-order-fairness (batch-OF)~\cite{Kelkar_Cornell_Aequitas_Aug2020} has emerged as the most practical yet meaningful choice, and it guarantees that whenever a fraction $\gamma$ of the replicas locally receive a transaction $\mathit{tx}$ before another transaction $\mathit{tx}'$, then all honest replicas must deliver $\mathit{tx}$ no later than $\mathit{tx}'$.
The Themis leader-based protocol~\cite{Kelkar_Cornell_Themis_Nov2022} demonstrated that batch-OF can be retrofitted onto HotStuff~\cite{hotstuff} with standard liveness. However, leader-based OF protocols remain fundamentally capped by the single-leader bottleneck of their underlying consensus engine, where one replica must collect all local orderings, compute the expensive global fair ordering, and drive consensus.

DAG BFT protocols~\cite{Narwhal&Tusk_Mar2022, DAG_Rider_Jun2021, Bullshark_Sep2022, Shoal_2023, Shoal++_2025, Mysticeti_2025, Mahi-Mahi_2025} eliminate this bottleneck by allowing replicas to propose vertices concurrently, achieving substantially higher throughput. Beyond throughput, DAG BFT protocols also tolerate faulty leaders gracefully, as the DAG continues to grow at full rate even with faulty replicas, with only the commitment being temporarily delayed.
Two recent proposals have brought batch-OF to DAG-based BFT.
FairDAG-RL~\cite{Kang_FairDAG_Aug2025} constructs global dependency graphs \emph{after} each subdag is committed (post-consensus) and resolves missing edges \emph{implicitly}, using ordering indicators deposited in later committed subdags.
DoD~\cite{Nagda_UPenn_PhD_2025} takes the opposite approach, constructing the global dependency graph \emph{before} consensus and embedding it in the payload of each DAG vertex.
While both protocols improve upon Themis in raw throughput by leveraging the multi-proposer DAG design, each of them serializes sub-DAG execution that prevents them from approaching the throughput of their underlying DAG consensus.

\textbf{Problem.}
Although FairDAG-RL and DoD differ in where graph construction occurs, both protocols force their fairness layers into \textbf{strictly serial, round-by-round execution}, just at different points in the pipeline.
In DoD, the global-order graph is constructed as part of the DAG vertex creation process. This is an $O(B^2)$ operation (where $B$ is the batch size) that \textbf{sits on the critical path of every DAG round}. The DAG cannot advance to the next round until the fairness computation completes, directly throttling the DAG's otherwise fast round progression.
In effect, DoD inherits the high-throughput potential of DAG-based consensus only to immediately surrender it by embedding the most expensive fairness operation on the DAG's critical path.
In FairDAG-RL, graph construction happens after consensus, so the underlying DAG runs at full speed. However, the fairness layer still processes one subdag at a time, leaving the rest of the validator's CPU idle and forcing fair ordering to become the system bottleneck even when the DAG runs comfortably.

\noindent\textbf{Where the fairness layer actually spends its time.}
The fairness layer in both FairDAG-RL and DoD inherits the same high level structure as the Themis \textsc{FairPropose} algorithm.
Each committed subdag goes through four phases on the fairness processor: (i) extracting the per replica local orderings from the subdag's certificates, (ii) building the pairwise weight matrix that records how many replicas place each pair of transactions in each direction, (iii) running Tarjan and topological sort on the resulting dependency graph, and (iv) applying the result back to cumulative state.
To understand what makes serial execution so costly in practice, we measured per task CPU time across these four phases in FairDAG-RL and Figure~\ref{fig:exp_cpu_time} shows the result.
The pairwise weight matrix alone per committed subdag consumes more than three times the cost of Tarjan plus topological sort and almost an order of magnitude more than the extraction and result handling steps combined.
The weight phase is also the only phase whose cost scales quadratically with the size of the active vertex set, which itself grows with throughput.
A serial fairness layer that finalizes one subdag at a time therefore spends most of its wall clock time on the expensive weight matrix calculation and leaves the other CPU cores idle.
This observation, that the dominant cost is exactly the part of the algorithm that has no cross subdag data dependency, is what motivates Herring's central design choice.

\begin{figure}[t]
    \centering
    \includegraphics[width=\linewidth]{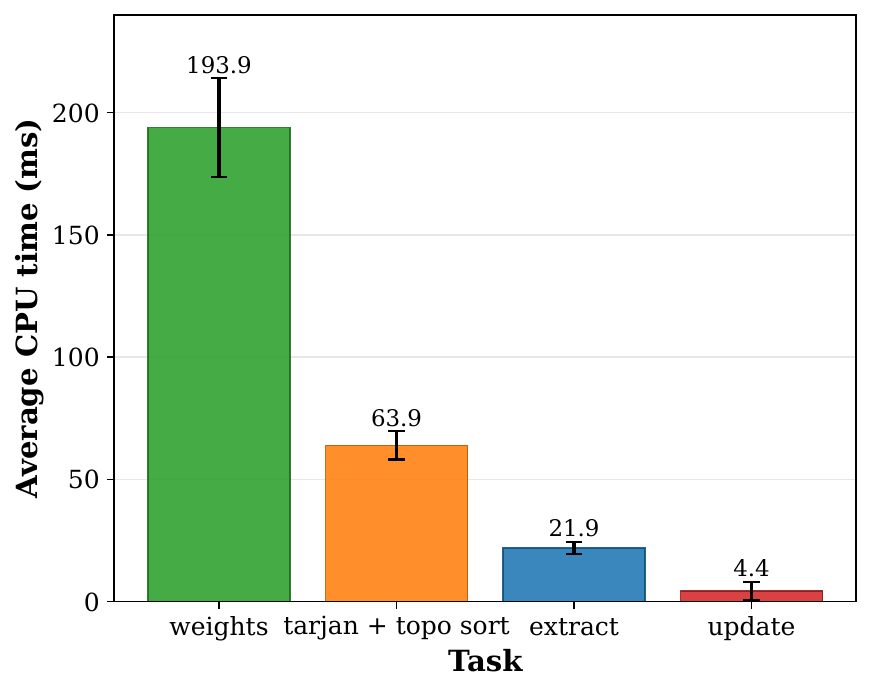}
    \caption{Average CPU time per phase of FairDAG-RL's fairness layer, broken down by the four \textsc{FairPropose} phases of extract local orderings, pairwise weight matrix computation, Tarjan plus topological sort, and result update. The pairwise weight phase dominates total cost.}
    \label{fig:exp_cpu_time}
\end{figure}

\noindent\textbf{Solution.}
In this paper we present Herring, a batch-OF DAG BFT protocol built on top of Narwhal \& Tusk~\cite{Narwhal&Tusk_Mar2022}, whose central design insight is that the dominant per-subdag graph construction cost can be made \emph{parallel} across committed subdags, with the remaining cross subdag work confined to a small number of synchronization points.
Herring exploits this insight by running graph construction tasks for multiple committed subdags concurrently on a thread pool, turning OF from a per-round serial bottleneck into a CPU-bound task that scales with the number of available validator cores.
Two supporting design choices keep this parallelism sound.
Like FairDAG-RL, Herring builds its dependency graphs post-consensus, so that all fairness work sits off the DAG's critical path.
Inspired by Themis~\cite{Kelkar_Cornell_Themis_Nov2022}, Herring resolves missing edges \emph{explicitly}\footnote{An implicit resolution scheme as in FairDAG-RL is in principle compatible with parallel construction, but it forces every parked graph to wait for evidence to trickle in through later committed subdags, with every subdag task needing to read and mutate the weights of earlier parked graphs, thus making concurrency complex to manage. Explicit missing edge resolution decouples each parked graph from the others and lets all of them accumulate votes concurrently through Narwhal's reliable broadcast.} via FairUpdate votes that piggyback on workers' outgoing batches and travel through Narwhal's reliable broadcast, so that each parked graph accumulates its resolutions independently.
Table~\ref{tab:protocol_comparison} summarizes the key design differences between Herring, FairDAG-RL and DoD.


\begin{table*}[t]
\centering
\caption{Design comparison of the leading batch-OF DAG BFT protocols. An entry in \good{green} indicates the best design choice, and an entry in \bad{red} indicates a shortcoming. Liveness attacks\textsuperscript{$\dagger$} were identified and patched in this work (Appendix~\ref{liveness_attacks}). The Corruption column reports the fault tolerance threshold at $\gamma = 1$. FairDAG-RL achieves $n \geq 3f{+}1$ because it assumes a DAG with weak edges\textsuperscript{$\ast$}, whereas Herring and DoD operate on DAG without weak edges and require $n \geq 4f{+}1$.}
\label{tab:protocol_comparison}
\renewcommand{\arraystretch}{1.35}
\begin{tabular}{l ccccc c}
\toprule
 & \makecell{Graph\\Construction}
 & \makecell{Missing Edge\\Resolution}
 & \makecell{Performance\\Bottleneck}
 & \makecell{Liveness\\Attacks}
 & \makecell{Corruption\\($\gamma = 1$)}
 & \makecell{Parallel\\Graph Construction} \\
\midrule
FairDAG-RL~\cite{Kang_FairDAG_Aug2025}
    & \good{Post-consensus}
    & \bad{Implicit}
    & \bad{Fairness layer}
    & \bad{Yes\textsuperscript{$\dagger$}}
    & $n \geq 3f{+}1$\textsuperscript{$\ast$}
    & \multicolumn{1}{|c}{\cellcolor{red!10}\bad{No}} \\
DoD~\cite{Nagda_UPenn_PhD_2025}
    & \bad{Pre-consensus}
    & \bad{Implicit}
    & \bad{DAG pipeline}
    & \bad{Yes\textsuperscript{$\dagger$}}
    & $n \geq 4f{+}1$
    & \multicolumn{1}{|c}{\cellcolor{red!10}\bad{No}} \\
Herring (this work)
    & \good{Post-consensus}
    & \good{Explicit}
    & \good{None}
    & \good{No}
    & $n \geq 4f{+}1$
    & \multicolumn{1}{|c}{\cellcolor{green!15}\good{Yes}} \\
\bottomrule
\end{tabular}

\vspace{0.5em}
\footnotesize
\textsuperscript{$\ast$} Weak edges are block references to the left-behind not-directly-referenced blocks, e.g., references to slow but honest nodes' blocks. With weak edges every honest node's local ordering is present upon subdag commit. Without weak edges this is not the case, thus the threshold becomes $n \geq 4f{+}1$. Note that Narwhal~\cite{Narwhal&Tusk_Mar2022} by design removes weak edges specifically to enable garbage collection, which is infeasible when weak edges are present.
\end{table*}


\textbf{Overview and Contributions.}
As a summary, this paper makes the following contributions:
    
    $\bullet$ We present Herring, the first batch-OF DAG BFT protocol that parallelizes global graph construction across committed subdags, by combining post-consensus graph construction with explicit missing edge resolution that piggybacks on Narwhal's existing batch dissemination.

    $\bullet$ We identify and analyze previously unreported liveness bugs in the fairness layers of both FairDAG-RL~\cite{Kang_FairDAG_Aug2025} and DoD~\cite{Nagda_UPenn_PhD_2025}. For each bug we give a concrete trigger scenario, propose a patch, and integrate the patch into our reimplementations so that both baselines actually run to completion in our evaluation (Appendix~\ref{liveness_attacks}).

    $\bullet$ We prove that Herring satisfies $\gamma$-batch-OF, together with the standard safety and liveness properties of its underlying DAG BFT consensus.

    $\bullet$ We implement Herring on top of the Rust implementation of Narwhal \& Tusk~\cite{Narwhal&Tusk_Mar2022} and evaluate it against FairDAG-RL, DoD-W, and Themis across a range of workloads, network sizes, and fairness parameters. Herring tracks the throughput of its underlying DAG consensus closely up to $10{,}000$\,tx/s, sustains roughly $90\%$ higher saturation throughput than FairDAG-RL and $100\%$ higher than DoD-W, and achieves up to $75\%$ lower execution latency than both DAG-based baselines.

    $\bullet$ We demonstrate Herring's structural resistance to Byzantine collusion by comparing it against Themis under the reversing order adversarial strategy of Kelkar et al.~\cite{Kelkar_Cornell_Themis_Nov2022}, showing that Herring confines adversarial reordering to only the most fragile transaction pairs.


\section{Background}

This section introduces the two foundations Herring builds on: $\gamma$-batch-OF, the fairness property Herring guarantees, and Narwhal \& Tusk, its underlying DAG BFT consensus.

\subsection{\texorpdfstring{$\gamma$-Batch-Order-Fairness}{gamma-Batch-Order-Fairness}}

Traditional BFT consensus protocols guarantee safety and liveness but place no constraint on the relationship between the order in which transactions are received by replicas and the order in which they appear in the total order. Order-fairness protocols aim to enforce such a relationship.

The strongest intuitive notion is \emph{receive-order-fairness}~\cite{Kelkar_Cornell_Aequitas_Aug2020}, which requires that if a $\gamma$ fraction of replicas receive $\mathit{tx}$ before $\mathit{tx}'$, then all honest replicas must output $\mathit{tx}$ strictly before $\mathit{tx}'$. However, Kelkar et al.~\cite{Kelkar_Cornell_Aequitas_Aug2020} showed that receive-order-fairness is impossible in general, even when all replicas are honest, due to the \textit{Condorcet paradox}. As a simple example, consider three replicas receiving transactions in the orders $[a, b, c]$, $[b, c, a]$, and $[c, a, b]$. A majority prefers $a$ before $b$, $b$ before $c$, and $c$ before $a$, forming a cycle that cannot be linearized without violating at least one majority preference.

To circumvent this impossibility, Kelkar et al.~\cite{Kelkar_Cornell_Aequitas_Aug2020} introduced the relaxed notion of $\gamma$\emph{-batch-OF} (see Definition~\ref{def:batch_of}), which groups cyclically dependent transactions into batches and only enforces ordering constraints across batches. Transactions belonging to the same batch are output at the 'same time'. Further, within a batch, any total ordering is allowed, as the fairness guarantee only constrains relative ordering across batches.

\begin{definition}[$\gamma$-Batch-Order-Fairness~\cite{Kelkar_Cornell_Aequitas_Aug2020}]
\label{def:batch_of}
For any two transactions $\mathit{tx}$ and $\mathit{tx}'$ received by all replicas, if $\gamma n$ replicas receive $\mathit{tx}$ before $\mathit{tx}'$, then all honest replicas output $\mathit{tx}$ no later than $\mathit{tx}'$, where $\tfrac{1}{2} < \gamma \leq 1$.
\end{definition}

Aequitas~\cite{Kelkar_Cornell_Aequitas_Aug2020} first realized batch-OF only guarantees weak liveness, because Condorcet cycles can chain together and grow to arbitrary length. Themis~\cite{Kelkar_Cornell_Themis_Nov2022} solved this through \emph{batch unspooling}, which allows transactions within a cycle to be output incrementally without waiting for the cycle to fully form, while still ensuring that all transactions in the same batch appear contiguously in the output. Herring adopts the same unspooling approach.

\smallskip
\noindent\textbf{Dependency graphs.}
Both Themis and the batch-OF DAG BFT protocols that followed construct a \emph{dependency graph} to capture pairwise ordering preferences. Given $n{-}f$ replica orderings, a directed edge from $\mathit{tx}$ to $\mathit{tx}'$ is added when the number of replicas that received $\mathit{tx}$ before $\mathit{tx}'$ exceeds a threshold of $n(1{-}\gamma) + f + 1$. Transactions are classified by the number of ordering reports they have accumulated: (i) \emph{Solid}, if the transaction appears in at least $n{-}2f$ orderings; (ii) \emph{Shaded}, if it appears in at least $n(1{-}\gamma) + f + 1$ orderings but in fewer than $n{-}2f$; and (iii) \emph{Blank}, otherwise.

Only non-blank transactions are inserted into the dependency graph. The ordering is finalized once the graph becomes a \emph{tournament}, i.e., when every pair of non-blank vertices is connected by exactly one directed edge. Pairs that have not yet accumulated enough evidence to determine a direction are called \emph{missing} and their edges must be resolved before finalization.

\subsection{Narwhal \& Tusk}
\label{sec:narwhal_tusk}

Narwhal \& Tusk~\cite{Narwhal&Tusk_Mar2022} is the asynchronous DAG BFT protocol Herring builds on. It decouples transaction dissemination from ordering, eliminating the leader bandwidth bottleneck of traditional BFT protocols. Each replica (called a \emph{validator}) runs a \emph{primary} process and one or more \emph{worker} processes. Workers receive client transactions, assemble them into batches, and reliably broadcast those batches to workers with the same identifier on other validators. The primary manages only lightweight metadata and the DAG structure, referencing worker batches by their cryptographic digests. This separation allows throughput to scale horizontally by adding more workers per validator.

\smallskip
\noindent\textbf{DAG construction.}
The protocol proceeds in rounds. In each round $r$, every validator proposes a single DAG vertex through its primary, containing references (strong edges) to at least $n{-}f$ vertices from round $r{-}1$. Before a vertex can be proposed, the primary must collect $n{-}f$ certificates of availability from the previous round, where each certificate attests that at least $f{+}1$ honest validators have stored the corresponding vertex's data. This reliable broadcast mechanism ensures that any certified vertex is retrievable even if its original proposer fails. Unlike DAG-Rider~\cite{DAG_Rider_Jun2021}, Narwhal \& Tusk does not use weak edges (i.e., references to vertices emitted before round $r{-}1$), which allows replicas to garbage collect older uncommitted (orphaned) DAG vertices, thus making the protocol practical for deployment.

\smallskip
\noindent\textbf{Consensus and total ordering.}
Tusk operates on top of the Narwhal DAG without additional communication. It groups rounds into \emph{waves} of fixed length, at the end of which a leader vertex $L_r$ is elected through a shared coin. If $L_r$ receives at least $f{+}1$ references in the following round, it is committed. Upon commitment, all vertices in the causal history of $L_r$ not already committed under an earlier leader form the \emph{subdag} $A_r$. The DAG is thereby partitioned into non-overlapping subdags $(A_{r_1}, A_{r_2}, \ldots)$. Vertices of each subgraph $A_r$ are deterministically ordered by topological sort.

\smallskip
\noindent\textbf{Key properties.}
Narwhal \& Tusk guarantees the following properties:
\begin{itemize}
    \item \emph{Agreement}. If a correct replica commits a vertex $v$, then all correct replicas eventually commit $v$.
    \item \emph{Total order}. If a correct replica commits $v$ before $v'$, then every correct replica commits $v$ before $v'$.
    \item \emph{Validity}. If a correct replica broadcasts a vertex $v$, then all correct replicas eventually commit $v$.\footnote{As Narwhal's garbage collection might orphan vertices that never commit, it is then necessary to re-inject the transactions of orphaned vertices into later vertices, which together with the $1/2$-Chain Quality property of the DAG ensures eventual commitment within a constant number of rounds in expectation.}
\end{itemize}

\smallskip
\noindent\textbf{Self-referencing rule.}
In addition to the standard vertex validity rules, Herring requires that each correct replica includes a reference to its own previous round certificate when proposing a new vertex.
Without this rule, a correct replica's earlier vertex can be orphaned while a later vertex is committed, causing ordering evidence to arrive out of order across subdags.
The self-referencing chain prevents this by guaranteeing that if a replica's round $r$ vertex is committed, all of its earlier vertices are in the causal history and therefore committed in the same or an earlier subdag.
Combined with the worker's monotonic \emph{local ordering indicator (LOI)} assignment, this ensures that each correct replica's cumulative ordering across committed subdags is non-decreasing in LOI.
FairDAG~\cite{Kang_FairDAG_Aug2025} adopts the same rule.

\section{System Model}

This section describes the system and threat model under which Herring operates. Herring inherits the asynchronous network and Byzantine fault assumptions of its underlying DAG BFT protocol, Narwhal \& Tusk~\cite{Narwhal&Tusk_Mar2022}, and strengthens only the fault-tolerance threshold to match the requirement of $\gamma$-batch-OF.

\smallskip
\textbf{System and threat model.}
We consider a system of $n$ known replicas communicating via message passing over an asynchronous network where messages between correct replicas are eventually delivered but with no bound on delay. A computationally bounded adversary can corrupt up to $f$ replicas during execution. Corrupted replicas are Byzantine and may behave arbitrarily, e.g., they can reorder or withhold messages, and report false local orderings. The remaining $n{-}f$ replicas follow the protocol faithfully. Replicas are connected by pairwise authenticated channels. Clients submit transactions to all replicas and wait for finalization.

Following the standard assumption in OF protocols~\cite{Kelkar_Cornell_Aequitas_Aug2020,Kelkar_Cornell_Themis_Nov2022,Kang_FairDAG_Aug2025}, the external network between clients and replicas is non-adversarial, as an adversary controlling transaction delivery could manipulate input orderings directly and render any OF guarantee vacuous.

\begin{figure*}[t]
    \centering
    \begin{subfigure}[c]{0.48\textwidth}
        \centering
        \includegraphics[width=\linewidth]{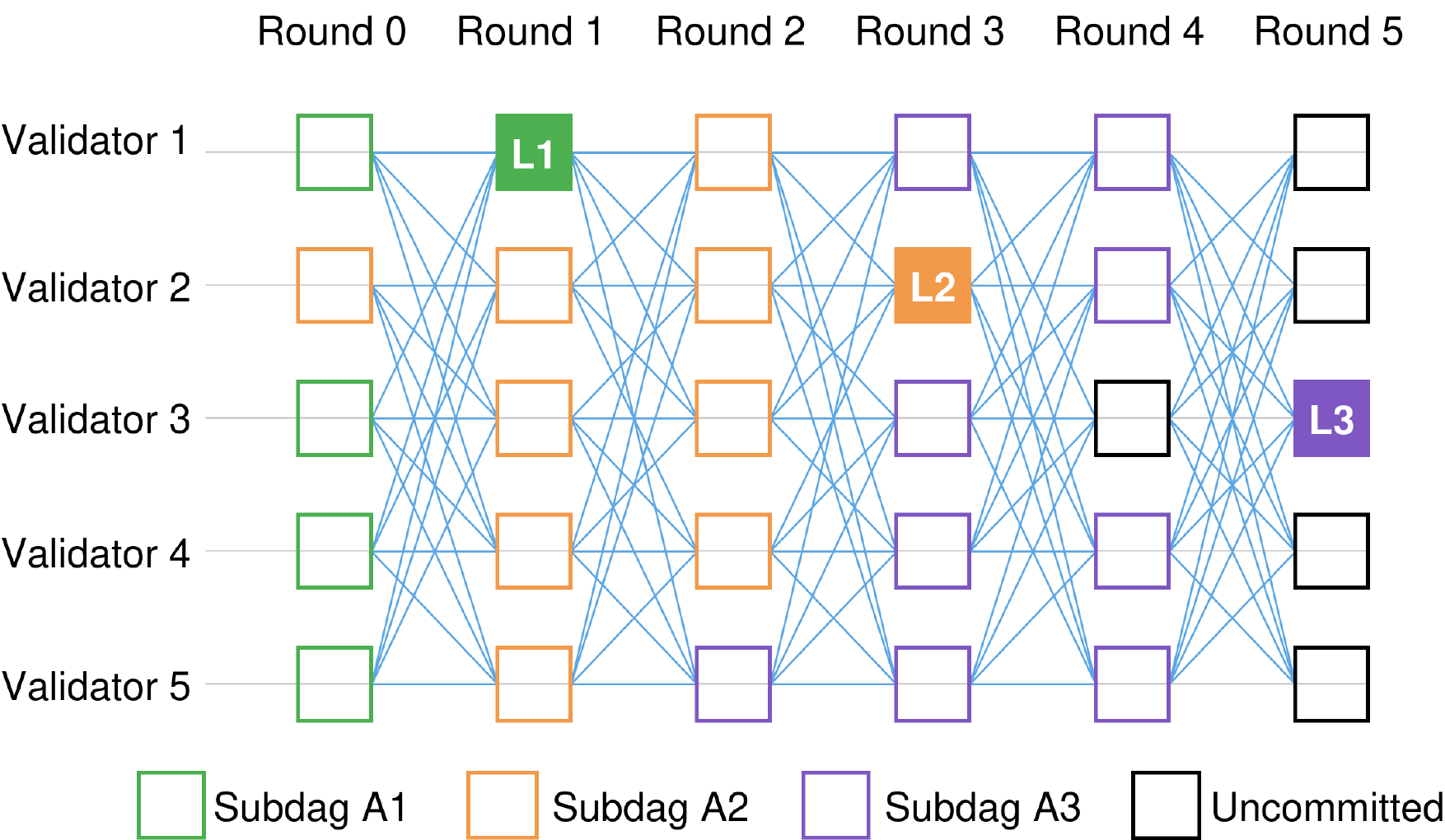}
        \caption{DAG BFT consensus layer. An unmodified Narwhal \& Tusk
        DAG with five validators across six rounds, with committed leaders
        $L_1, L_2, L_3$ producing subdags $A_1, A_2, A_3$.}
        \label{fig:herring_arch_a}
    \end{subfigure}\hfill
    \stepcounter{subfigure}
    \begin{subfigure}[c]{0.48\textwidth}
        \centering
        \includegraphics[width=\linewidth]{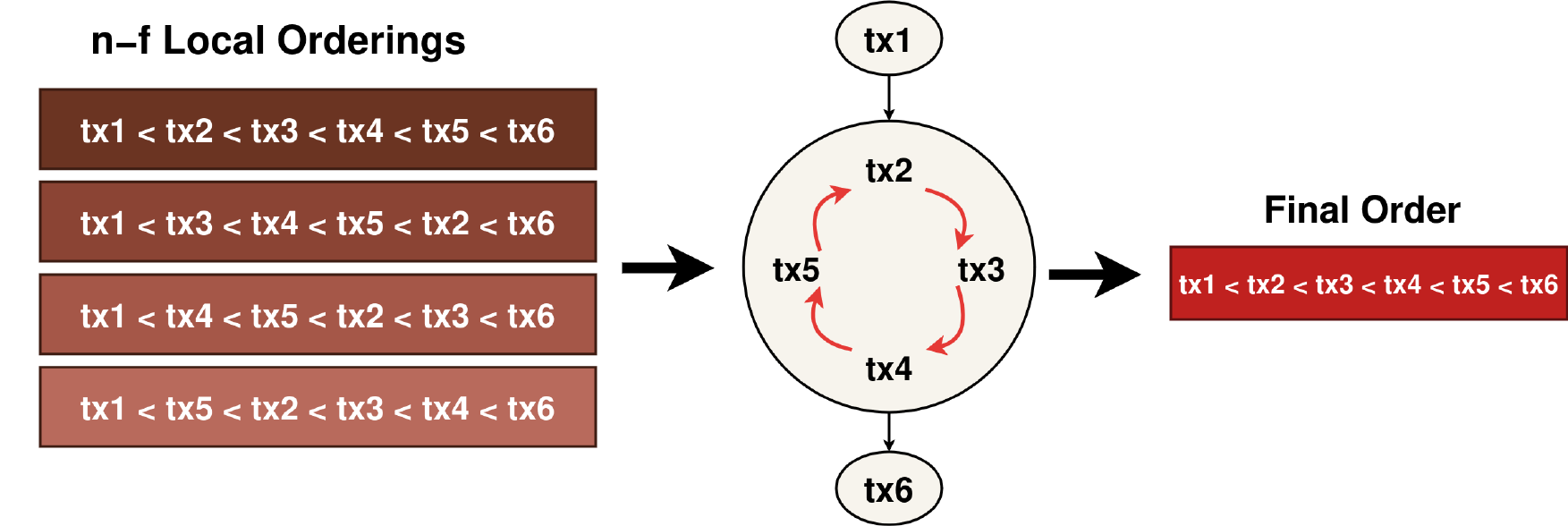}
        \caption{Per-subdag graph and order resolution. From the $n{-}f$
        local orderings, the task builds $G_r$, decomposes it into SCCs,
        identifies the anchor, and linearizes each SCC up to the anchor.
        Transactions forming a Condorcet cycle
        ($\{\mathit{tx}_2, \mathit{tx}_3, \mathit{tx}_4, \mathit{tx}_5\}$)
        are output `at the same time' as part of the same batch.}
        \label{fig:herring_arch_c}
    \end{subfigure}

    \vspace{0.6em}

    \setcounter{subfigure}{1}
    \begin{subfigure}{\textwidth}
        \centering
        \includegraphics[width=\linewidth]{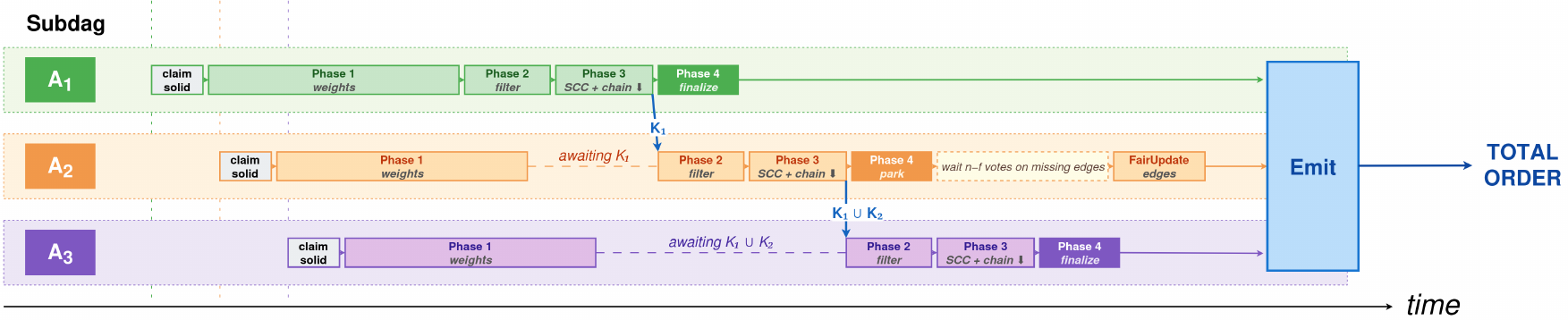}
        \caption{Herring per-subdag parallel pipeline. After a synchronous
        extract\,+\,solid-claim step on the main thread, each subdag's
        task runs four phases on a thread pool. \emph{Phase~1} (pairwise
        weight matrix, the dominant cost) is fully parallel across
        in flight subdags. The only cross subdag synchronization is the
        cumulative-$K$ chain forwarded early by \emph{Phase~3}, which the
        next subdag's \emph{Phase~2} awaits before filtering its active
        set. \emph{Phase~4} either finalizes or parks the graph, where
        parked subdags accumulate $n{-}f$ FairUpdate votes over Narwhal's
        reliable broadcast and then run \textsc{ApplyFairUpdate}. The
        Emit serialization point drains ready buffers in subdag commit
        order. Dashed wait segments on $A_2, A_3$ visualize the passing
        of the cumulative chain, i.e., $A_2$'s \emph{Phase~2} cannot
        begin until $K_1$ arrives from $A_1$'s \emph{Phase~3}, and
        $A_3$'s until $K_1 \cup K_2$ arrives from $A_2$'s.}
        \label{fig:herring_arch_b}
    \end{subfigure}

    \caption{Herring system overview.}
    \label{fig:herring_architecture}
\end{figure*}

\smallskip
\textbf{Fault-tolerance threshold.}
Narwhal \& Tusk requires $n \geq 3f{+}1$ for consensus safety and liveness. Herring strengthens this to the threshold required by $\gamma$-batch-OF, i.e., $n > \frac{4f}{2\gamma - 1}$, where $\frac{1}{2} < \gamma \leq 1$, reducing to $n \geq 4f{+}1$ for $\gamma = 1$ ($n \geq 3f{+}1$ if only crash faults are present). This threshold is a direct consequence of building on top of Narwhal \& Tusk, whose DAG does not contain weak edges. Narwhal achieves its validity property through reinjection of orphaned vertices into later rounds, which eventually commits every correct replica's transactions but provides no bound on how many rounds this takes. For fair ordering, waiting for reinjection is not an option, since the dependency graph for a subdag must be built from the orderings actually present in that committed subdag. A committed subdag contains at most $n{-}f$ orderings, of which up to $f$ may be Byzantine, leaving $n{-}2f$ honest contributions. The Themis threshold of $n > \frac{4f}{2\gamma-1}$ is what makes this honest majority strict for every transaction pair, even after Byzantine influence. 

Protocols built on DAGs with weak edges, such as FairDAG-RL, can achieve a weaker threshold of $n \geq 3f{+}1$ for $\gamma = 1$ because weak edges add $f$ extra honest contributions per subdag. However, weak edges come at a significant practical cost. As noted by Danezis et al.~\cite{Narwhal&Tusk_Mar2022}, the original Narwhal \& Tusk paper explicitly removes weak edges precisely because they make garbage collection infeasible, since any block received may eventually be referenced by a future weak edge and must therefore be stored indefinitely. Herring inherits Narwhal \& Tusk's garbage collection and operates within a fixed memory footprint, at the cost of the stronger $n \geq 4f{+}1$ threshold.

Following the analysis of Ladelsky et al.~\cite{Ladesky_2025}, accommodating this stronger threshold requires adjusting the underlying Narwhal protocol so that rounds consist of at least $(k{-}1)f + 1$ vertices and each vertex references $(k{-}1)f + 1$ vertices from the previous round, where $k = \frac{4}{2\gamma-1}$ for Byzantine faults (reducing to $k \geq 3$ for crash faults only).

\smallskip
\textbf{Cryptography.}
We assume a public-key infrastructure (PKI) where each replica holds a unique private key and all public keys are known. Digital signatures allow any replica to verify the origin and integrity of a message. We further assume a collision-resistant hash function $H$ that maps arbitrary-length inputs to fixed-length digests, used throughout the protocol to identify transactions and batches~\cite{KatzLindell2014}.

\section{Herring: Parallel Batch-OF on DAG BFT}
\label{sec:herring}
 
This section describes the design of Herring.
We first give an architectural overview, then describe how workers assign local ordering indicators (LOIs) and cast FairUpdate votes, and finally present the per-subdag graph construction and order finalization.

\subsection{Overview}
\label{sec:herring_overview}

Herring preserves the architecture of Narwhal \& Tusk and adds a per replica \emph{fairness processor} that builds a global dependency graph and emits a fair transaction order upon every committed subdag.
The starting point of Herring's design is the observation from the introduction that the dominant cost in any Themis-style fairness layer is the pairwise weight matrix computation (phase 1 of \textsc{FairPropose}), and that this cost has no cross subdag data dependency. Herring exploits this by running phase 1 of multiple committed subdags concurrently on a thread pool, and serializing only the parts that actually need cross subdag agreement.

Three design choices make per subdag processing safely parallel.
First, dependency graphs are constructed \emph{post-consensus}, so the fairness processor runs only after Tusk commits a subdag, keeping all fairness work off the DAG's critical path.
Second, missing edges are resolved \emph{explicitly} via FairUpdate votes that piggyback on workers' outgoing batches and travel through Narwhal's reliable broadcast.
Each subdag's graph receives its resolutions directly rather than waiting for implicit evidence from later subdags.
Third, support counts and pairwise weights are \emph{cumulative} across committed subdags, with two complementary mechanisms ensuring that each transaction enters at most one graph despite parallel processing. A solid claim recorded synchronously at dispatch time excludes solid transactions of any in flight subdag from later snapshots, and a cumulative chain forwarded along the parallel task pipeline excludes shaded transactions of any in flight subdag from later active sets.

Figure~\ref{fig:herring_architecture} illustrates the high level architecture
across three sub-figures.
Sub-figure~\ref{fig:herring_arch_a} shows the unmodified DAG BFT layer that
commits subdags $A_1, A_2, A_3, \dots$.
Sub-figure~\ref{fig:herring_arch_b} shows the Herring per subdag pipeline that
runs after each commit.
Sub-figure~\ref{fig:herring_arch_c} shows the per subdag graph $G_r$ that the
pipeline produces and from which the per subdag transaction order is
extracted.
 
Workers assign each transaction a monotonically increasing LOI on first
observation and broadcast batches containing transactions, LOIs, and any
pending FairUpdate votes.
Primaries assemble the DAG and Tusk commits subdags exactly as in vanilla
Narwhal \& Tusk.
Each committed subdag is then handed to the fairness processor, which
performs lightweight synchronous bookkeeping on the main thread before
dispatching a four phase task onto a thread pool
(Sub-figure~\ref{fig:herring_arch_b}).
\emph{Phase~1} computes the support counts and pairwise weight matrix.
It is the dominant cost (Figure~\ref{fig:exp_cpu_time}) and runs fully in
parallel across in flight subdags.
\emph{Phase~2} awaits the cumulative $K$ chain from the prior subdag and uses
it as an active set filter before building $G_r$.
\emph{Phase~3} runs Tarjan SCC decomposition and anchor truncation, then
forwards the extended chain to the next subdag's task \emph{before} doing any
missing edge work, so that $A_{r+1}$'s Phase~2 can begin while $A_r$ is still
in Phase~4.
\emph{Phase~4} either returns a finalized order or parks the graph and
triggers a FairUpdate vote exchange.
Three synchronization points discipline this design.
First, a solid claim recorded synchronously before dispatch excludes solid
transactions of any in flight subdag from later snapshots.
Second, the cumulative $K$ chain handoff between consecutive subdags excludes
shaded transactions of any in flight subdag from later active sets.
Third, the Emit point drains ready buffers in subdag commit order, so that
the transaction order of $A_{r-1}$ is always emitted before that of $A_r$,
ensuring all correct replicas produce the same total order.
 
\subsection{Worker Batch Construction and FairUpdate Voting}
\label{sec:herring_batch}
 
Each Narwhal worker maintains a LOI tracker that assigns a monotonically increasing counter value to every transaction upon first observation, regardless of whether the transaction arrives directly from a client or inside a remote worker's batch. \footnote{Our implementation runs the LOI tracker inside a single worker per validator, which keeps the local receive order centralized and avoids cross worker LOI reconciliation. A natural extension is to dedicate one worker per validator as the \emph{OF worker} that exclusively handles LOI assignment and FairUpdate voting, while other workers retain their normal Narwhal duties. The OF worker then serves as the validator's authoritative local receive order and feeds LOIs to the rest of the system over a side channel. We leave this engineering split to future work.}
Subsequent observations of the same transaction return the previously assigned LOI.
This guarantees that each (replica, transaction) pair has a stable LOI reflecting the true first observation time, independent of the ingress path.

\begin{algorithm}[H]
\small
\caption{FairUpdate voting at the worker}
\label{alg:herring_voting}
\begin{algorithmic}[1]
\State \textbf{State:} LOI tracker $T$, unresolved edges $P[r]$, resolved votes $V[r]$
\Statex
\State \textbf{upon} $\mathit{FairPropose}(r, M_r)$ from fairness processor \textbf{do}
\For{each $(u, v) \in M_r$}
    \If{$T$ knows $\mathrm{loi}(u)$ and $\mathrm{loi}(v)$}
        \State record directed vote $(\min \!\to\! \max)$ by LOI in $V[r]$
    \Else
        \State $P[r] \gets P[r] \cup \{(u, v)\}$
    \EndIf
\EndFor
\Statex
\State \textbf{upon} new transaction $\mathit{tx}$ observed \textbf{do}
\State $T.\textsc{Record}(\mathit{tx})$
\For{each $r$ and each $(u, v) \in P[r]$ with $\mathit{tx} \in \{u, v\}$}
    \If{$T$ knows $\mathrm{loi}(u)$ and $\mathrm{loi}(v)$}
        \State record directed vote in $V[r]$, remove $(u,v)$ from $P[r]$
    \EndIf
\EndFor
\Statex
\State \textbf{upon} $P[r] = \emptyset$ for some $r$ \textbf{do}
\State \quad queue $V[r]$ into next outgoing batch, clear $V[r]$
\end{algorithmic}
\end{algorithm}
 
A batch carries three kinds of payload alongside its raw transactions.
\emph{Direct entries} pair each newly received client transaction with its LOI.
\emph{Indirect entries} pair each transaction digest learned from a remote batch with its local LOI and are propagated exactly one hop.
\emph{FairUpdate votes} contain directed edge resolutions for previously parked subdags.
 
When the fairness processor parks a subdag~$r$ with unresolved missing edges $M_r$, it sends a $\mathit{FairPropose}(r, M_r)$ message to the local worker.
Algorithm~\ref{alg:herring_voting} describes how the worker resolves each missing edge by comparing the LOIs of its two endpoints.
If both LOIs are known, the worker records a directed vote.
Otherwise, it defers the edge until the missing endpoint arrives via a future client submission or remote batch.
Once all edges in $M_r$ have been resolved, the worker queues the complete set of directed votes for subdag~$r$ into the next outgoing batch, so that edge resolutions travel through Narwhal's existing reliable broadcast without additional communication rounds.
 
\subsection{Per-Subdag Graph Construction}
\label{sec:herring_graph}

When Tusk commits a subdag $A_r$, the fairness processor first runs a short synchronous step on the main thread that maintains cumulative per replica orderings of all pending (not yet proposed) transactions, ingests $A_r$'s new contributions into them, and builds a snapshot $\{L_i\}$ for the parallel graph construction task.
Cumulative state is necessary because a naive per subdag scheme, where support counts are computed only from orderings within the current $A_r$, would fail to leverage the evidence that builds up across subdags.
Under realistic asynchrony, a transaction can be reported by enough replicas to cross the solid threshold $\tau_s$ only when their committing vertices are spread across several committed subdags.

\smallskip
\noindent\textbf{Single graph mechanism: solid claim plus cumulative chain.}
The single graph property, that each transaction enters at most one $G_r$, is the safety invariant that makes parallel execution sound.
Maintaining it across in flight tasks is the central challenge of the design, because the synchronous main thread cannot wait for any task to finish before extracting the next subdag's snapshot.
Herring uses two complementary mechanisms to handle this.

The first is a \emph{solid claim} recorded synchronously on the main thread before $A_r$'s task is dispatched. The claim is the set of solid transactions in $A_r$'s snapshot, and subsequent subdags exclude this set from their snapshots at extract time. Solid transactions are the right thing to claim eagerly because they are guaranteed to enter $G_r$ by the anchor rule of Themis' \textsc{FairPropose} algorithm~\cite{Kelkar_Cornell_Themis_Nov2022}. They have support at least $\tau_s$, so they are necessarily in some SCC at or before the anchor in topological order, regardless of the rest of the graph structure. The claim is dropped when $A_r$'s task returns to the main thread, and any transaction that actually entered $G_r$ is then promoted to permanent excluded state.

The second mechanism is a \emph{cumulative chain} that handles shaded transactions. Shaded transactions are not claimed because, unlike solids, they are not guaranteed to end up in $G_r$. A shaded transaction can have a path to a solid in $G_r$ and be retained, or it can be truncated past the anchor and need to reappear in a later subdag's snapshot. Claiming shaded transactions eagerly would either be wrong, since the truncated ones have to come back, or it would force every later subdag to wait for $A_r$'s task to finish, which would defeat the parallelism. Instead, each subdag's task forwards along the cumulative set $K_1 \cup K_2 \cup \dots \cup K_{r}$ of vertices retained in earlier graphs, and the next subdag's task uses this set as an active set filter inside its own post weight computation phase, i.e., shaded transactions of an in flight subdag $A_r$ therefore stay visible to phase 1 of $A_{r+1}, A_{r+2}, \dots$ during weight matrix construction, but they are excluded from the active set before Tarjan runs, so they never enter a second graph.

\smallskip
\noindent\textbf{Why this split.}
Splitting the single graph mechanism into a solid claim plus a cumulative chain is what lets Herring keep the weight phase fully parallel while still serializing only what has to be serialized.
Solid transactions account for the bulk of throughput at steady state, and excluding them at extract time means later subdags do not have to compute weights involving solids of in flight subdags at all.
Shaded transactions are visible to the weight phase of later subdags because their resolution depends on cross subdag evidence anyway, but the cumulative chain filter ensures they are dropped before Tarjan runs.
The result is that the most expensive operation, which Figure~\ref{fig:exp_cpu_time} identifies as the pairwise weight matrix, runs concurrently across in flight subdags, and the only synchronization point is the cumulative chain handoff that adds a single set union per subdag.

\begin{algorithm}[H]
\small
\caption{Per-subdag graph construction (parallel task)}
\label{alg:herring_graph_new}
\begin{algorithmic}[1]
\State \textbf{Input:} subdag id $r$, snapshot orderings $\{L_i\}$, oneshot receiver \textit{prior\_cumulative\_rx}, oneshot sender \textit{cumulative\_tx}
\State \textbf{Output:} finalized order or parked graph with missing edges $M_r$
\State \textbf{Thresholds:} $\tau = n(1{-}\gamma) + f + 1$, \enspace $\tau_s = n - 2f$
\Statex
\State \emph{// Phase 1: support classification and pairwise weight matrix.}
\For{each $\mathit{tx}$ appearing in any $L_i$}
    \State $c(\mathit{tx}) \gets |\{i : \mathit{tx} \in L_i\}|$
\EndFor
\State $V_r \gets \{\mathit{tx} : c(\mathit{tx}) \geq \tau\}$, \enspace
       $S_r \gets \{\mathit{tx} : c(\mathit{tx}) \geq \tau_s\}$
\For{each unordered pair $\{u, v\} \subseteq V_r$}
    \State $w_{uv} \gets |\{i : u \prec v \text{ in } L_i\}|$, \enspace
           $w_{vu} \gets |\{i : v \prec u \text{ in } L_i\}|$
\EndFor
\Statex
\State \emph{// Phase 2: cumulative chain barrier, then build $G_r$.}
\State \textit{prior\_cumulative} $\gets$ \textbf{await} \textit{prior\_cumulative\_rx}
\State $V_r \gets V_r \setminus \textit{prior\_cumulative}$
\State $G_r \gets (V_r, \emptyset)$, \enspace $M_r \gets \emptyset$
\For{each unordered pair $\{u, v\} \subseteq V_r$}
    \Comment{using cached $w$ from Phase 1}
    \If{$\max(w_{uv}, w_{vu}) \geq \tau$}
        \State add edge $(u,v)$ if $w_{uv} \geq w_{vu}$, else $(v,u)$
    \Else \enspace $M_r \gets M_r \cup \{(u, v)\}$
    \EndIf
\EndFor
\Statex
\State \emph{// Phase 3: SCC decomposition, anchor truncation, chain forward.}
\State $[C_1, \ldots, C_s] \gets \textsc{TopoSort}(\textsc{TarjanSCC}(G_r))$
\State $a \gets \max\{j : C_j \cap S_r \neq \emptyset\}$
\State truncate $G_r$ and $M_r$ to vertices in $C_1 \cup \cdots \cup C_a$
\State $K_r \gets$ vertices retained in $G_r$ after truncation
\State send $\textit{prior\_cumulative} \cup K_r$ on \textit{cumulative\_tx}
\Statex
\State \emph{// Phase 4: finalize or park.}
\If{$M_r = \emptyset$}
    \State \Return \textsc{Finalize}($G_r$, $a$) \Comment{Alg.~\ref{alg:herring_finalize_new}}
\Else
    \State park $(G_r, M_r, a)$ and send $\mathit{FairPropose}(r, M_r)$ to local worker
\EndIf
\end{algorithmic}
\end{algorithm}

When $A_r$'s task completes, transactions that entered $G_r$ are made permanent and removed from pending state, while transactions that were in $A_r$'s snapshot but did not enter $G_r$ stay in pending state and are eligible for the next subdag's snapshot.

\smallskip
\noindent\textbf{Per subdag task.}
Algorithm~\ref{alg:herring_graph_new} gives the parallel graph construction
in full.
The four phases were introduced in \S\ref{sec:herring_overview} and are
visualized in Sub-figure~\ref{fig:herring_arch_b}.

In phase 1, the task classifies each transaction by its support count $c(\mathit{tx})$, the number of distinct $L_i$ in which it appears.
Only transactions with $c \geq \tau$ are admitted to $V_r$.
A transaction with $c \geq \tau_s$ is additionally marked as solid.
For each pair of admitted vertices, the task counts how many snapshot orderings place one before the other. These pairwise counts are cached for use in phase 2.

In phase 2, after the cumulative chain barrier, the task adds a directed edge between two admitted vertices if the larger of the two cached counts reaches the shaded threshold $\tau$, and records the pair as a missing edge otherwise.
In phase 3, the task computes the strongly connected components (SCCs) of $G_r$ via Tarjan's algorithm, topologically sorts them, and identifies the \emph{anchor} as the last SCC in topological order containing at least one solid vertex. Only SCCs up to and including the anchor are retained, and $K_r$ is the set of retained vertices.
Forwarding the extended chain at this point, before missing edge handling in phase 4, is what preserves parallelism. The next subdag's phase 2 can begin as soon as $K_r$ is known, without waiting for $A_r$'s graph to finalize or park.

On result arrival, the main thread runs a synchronous step that promotes the vertices entering $G_r$ to $\mathit{proposed\_txs}$, removes them from pending state, and drops the solid claim recorded for $A_r$.
 
\subsection{FairUpdate and Order Finalization}
\label{sec:herring_finalization}
 
The fairness processor extracts FairUpdate votes from every committed subdag and routes them to the parked graph for the target subdag they reference.
Once a parked graph for subdag~$r$ has accumulated votes from $n{-}f$ distinct replicas, a finalization task is spawned on the thread pool.

\begin{algorithm}[H]
\small
\caption{FairUpdate resolution and order finalization}
\label{alg:herring_finalize_new}
\begin{algorithmic}[1]
\State \textbf{Per-replica state:}
\State \quad $\mathit{parked}[r]$: parked graph $(G_r, M_r, a)$ for subdag $r$
\State \quad $\mathit{votes}[r]$: accumulated votes, author $\to$ directed edges
\State \quad $\mathit{ready}[r]$: completed order for subdag $r$
\State \quad $\mathit{next}$: next subdag id to emit
\Statex
\State \textbf{upon} committed subdag contains FairUpdate votes \textbf{do}
\State \quad route each vote to $\mathit{votes}[r']$ by target subdag $r'$
\State \quad \textbf{for} each $r'$ with $|\mathit{votes}[r']| \geq n{-}f$ \textbf{do}
\State \quad \quad spawn \textsc{ApplyFairUpdate}($r'$) on thread pool
\Statex
\Function{ApplyFairUpdate}{$r$}
    \State $(G_r, M_r, a) \gets \mathit{parked}[r]$
    \For{each $(u, v) \in M_r$}
        \State $w_{uv} \gets |\{a : (u \!\to\! v) \in \mathit{votes}[r][a]\}|$
        \State $w_{vu} \gets |\{a : (v \!\to\! u) \in \mathit{votes}[r][a]\}|$
        \If{$\max(w_{uv}, w_{vu}) \geq \tau$}
            \State add edge $(u,v)$ if $w_{uv} \geq w_{vu}$, else $(v,u)$
        \EndIf
    \EndFor
    \State $\mathit{ready}[r] \gets$ \textsc{Finalize}($G_r$, $a$)
    \State \textsc{Emit}()
\EndFunction
\Statex
\Function{Finalize}{$G, a$}
    \State $[C_1, \ldots, C_s] \gets \textsc{TopoSort}(\textsc{TarjanSCC}(G))$
    \State $\mathit{order} \gets [\,]$
    \For{$j = 1$ to $a$}
        \State append \textsc{LinearOrder}($C_j$) to $\mathit{order}$
    \EndFor
    \State \Return $\mathit{order}$
\EndFunction
\Statex
\Function{Emit}{} \Comment{serialization point}
    \While{$\mathit{ready}[\mathit{next}]$ exists}
        \State output $\mathit{ready}[\mathit{next}]$ to application log
        \State $\mathit{next} \gets \mathit{next} + 1$
    \EndWhile
\EndFunction
\end{algorithmic}
\end{algorithm}

Algorithm~\ref{alg:herring_finalize_new} describes vote application and order finalization. \textsc{ApplyFairUpdate} tallies, for each missing edge, the number of distinct replicas that voted in each direction and adds the majority supported edge whenever the larger tally reaches the shaded threshold~$\tau$.
Because votes ride on the same reliable broadcast that carries Narwhal's transaction batches, every vote cast by a correct replica is eventually delivered to every other correct replica.
Therefore, every parked graph eventually accumulates the $n{-}f$ votes needed to trigger finalization.

\textsc{Finalize} recomputes the SCCs of the augmented graph, topologically sorts them, and linearizes each SCC up to and including the anchor.
Transactions forming a Condorcet cycle are output contiguously, while inter SCC ordering follows the topological sort, following the condensation plus unspooling approach of Themis~\cite{Kelkar_Cornell_Themis_Nov2022}.
The intra SCC linearization can be instantiated either as a Hamiltonian path through the tournament induced by the SCC or deterministically by sorting on transaction digests.
Both choices preserve $\gamma$-batch-OF since the fairness guarantee only constrains the ordering \emph{across} SCCs.
 
\textsc{Emit} is the \emph{serialization point} where ready buffers are drained in subdag commit order, emitting the order of $A_r$ only after every earlier subdag has been emitted.
This is necessary because Tusk commits $A_{r-1}$ before $A_r$, and the total transaction order must respect this commitment sequence across all correct replicas.
Graph construction (Algorithm~\ref{alg:herring_graph_new}) and vote application remain fully parallel.

\section{Correctness}
\label{sec:correctness}
This section gives an insight into the correctness of the protocol with full proofs deferred to Appendix~\ref{sec:full_proofs}.

Herring's fairness processor runs after consensus and does not modify round progression, vertex creation, or leader election. The only change to the underlying DAG is the self-referencing rule of Section~\ref{sec:narwhal_tusk}, which strengthens validity without weakening safety or liveness. Herring inherits the agreement, total order, and validity properties of Narwhal \& Tusk directly, and adds three fairness-layer guarantees on top. 

\begin{theorem}[Agreement on fair order]
\label{thm:agreement_main}
All correct replicas emit the same total transaction order.
\end{theorem}

\begin{theorem}[$\gamma$-batch-order-fairness]
\label{thm:fairness_main}
For any two transactions $\mathit{tx}$ and $\mathit{tx}'$ received by all replicas, if $\gamma n$ replicas receive $\mathit{tx}$ before $\mathit{tx}'$, all correct replicas output $\mathit{tx}$ no later than $\mathit{tx}'$.
\end{theorem}

\begin{theorem}[Liveness]
\label{thm:liveness_main}
Every transaction submitted by a correct client is eventually output by every correct replica.
\end{theorem}

The central technical step is a reduction. Herring's parallel execution emits the same total order as a serial reference in which each subdag's graph construction task, including any subsequent ApplyFairUpdate resolution, runs to completion before the next begins. The three synchronization points of Section~\ref{sec:herring_overview}, namely the synchronous solid claim at dispatch, the cumulative $K$ chain handoff between consecutive subdags, and the Emit serialization point, together suffice to enforce this equivalence. Once the reduction is in place, agreement collapses to determinism of the serial reference, which is a deterministic function of the committed subdag sequence and the observed FairUpdate vote set.

The fairness argument is borrowed from Themis~\cite{Kelkar_Cornell_Themis_Nov2022} with the substitution of the leader's collected list $L$ by the per-subdag snapshot $\{L_i^r\}$ and of the leader proposal by subdag $A_r$'s dependency graph $G_r$. Directional safety rules out wrong-direction edges in any $G_r$ and in any ApplyFairUpdate resolution. The non-trivial case is when $\mathit{tx}$ and $\mathit{tx}'$ finalize in different subdags, where we show they belong to the same Condorcet cycle. LOI monotonicity is the one place where the Themis argument needs adaptation, and it rests on the self-referencing rule and the worker's monotonic LOI assignment.

Liveness follows from Narwhal's validity together with the fact that every parked graph eventually accumulates $n-f$ FairUpdate votes, since votes ride on the same reliable broadcast that carries transaction data.

\section{Experimental Evaluation}
We evaluate Herring by comparing its performance to FairDAG-RL~\cite{Kang_FairDAG_Aug2025} and DoD-W~\cite{Nagda_UPenn_PhD_2025}, two leading batch-OF DAG BFT protocols. In addition, we also evaluate the overhead of fair ordering of Herring relative to Narwhal \& Tusk~\cite{Narwhal&Tusk_Mar2022}, its underlying DAG consensus protocol. We also benchmark against Themis~\cite{Kelkar_Cornell_Themis_Nov2022}, the leading batch-OF protocol built on top of the leader-based consensus protocol of HotStuff~\cite{hotstuff}. 

We implemented Herring from the Rust research implementation of Narwhal~\footnote{https://github.com/asonnino/narwhal} by modifying the batch structure of each Narwhal worker to allow fair transaction ordering. This included keeping track of a node's local receive order together with appending a node's local receive order position with each transaction within the batch. To allow for parallel CPU-bound global graph construction, we have utilized the thread pool implementation provided with the Tokio~\footnote{https://tokio.rs} Rust library.

For comparing against DoD, we utilize its released code\footnote{https://github.com/HeenaNagda/DoD} with modifications. Upon running the code, we noticed that it differs from the protocol specification of its accepted conference paper. For this reason, the code was adjusted to conform to the protocol description pseudo-code. However, by doing so we spotted a liveness bug within the protocol description itself concerning the implicit edge update mechanism (Appendix~\ref{liveness_attacks}). Furthermore, due to the architectural design of DoD, there is a probability that implicitly added edges might need to be discarded after consensus. This makes the implicit missing edge resolution fix non-intuitive. For this reason, we have chosen to extend the implementation with explicit edge resolving, similar to how Themis~\cite{Kelkar_Cornell_Themis_Nov2022} and Herring handle missing edges.

DoD builds on Rashnu's~\cite{Nagda_UPenn_Rashnu_May2024} data-dependent order-fairness notion, which enforces fair ordering only between transactions that access overlapping data objects, with the original DoD evaluation reporting on two variants, DoD-R (read heavy) and DoD-W (write heavy). The latter behaves similarly to Themis, i.e., assumes dependency between each transaction. Herring, FairDAG-RL, and Themis do not implement data-dependent fairness and order all transactions through the fairness layer regardless of data dependency. Therefore, in our evaluation we use DoD-W, which is the configuration that actually exercises DoD's fairness layer on every transaction. We note that data-dependent fairness is an orthogonal optimization that could be layered onto any of the four protocols. As such we do not explore this direction in this paper and leave it for future work.

The released implementation of FairDAG-RL~\cite{Kang_FairDAG_Aug2025} is built on top of ResilientDB~\cite{resilientdb}, a framework that differs to runtimes of both Herring and DoD-W, which makes the apples to apples comparison between protocols non-trivial. Instead, we have implemented a prototype of the FairDAG-RL protocol starting from the pseudo-code description\footnote{https://github.com/randomUserGithub123/narwhal/tree/fairdag}. The main difference lies in the underlying DAG structure. The original FairDAG-RL protocol assumes a DAG with \textit{weak edges}, which lets it count $f$ additional contributions per subdag and in turn reach the weaker fault tolerance threshold of $n \geq 3f{+}1$ for $\gamma = 1$. Since Herring, DoD-W, and Narwhal \& Tusk all operate on the weak-edge free Narwhal, running the original FairDAG-RL directly would compare protocols across two different DAG structures at once, which would conflate the effect of the DAG with the effect of the fairness layer. We therefore ported FairDAG-RL to run on plain Narwhal, dropping the extra $f$ contributions that weak edges would provide and raising its threshold to match the other protocols ($n \geq 4f{+}1$ for $\gamma = 1$). This gives a fair comparison of fairness layers on the same DAG substrate. We note that this comes at a cost for FairDAG-RL, since its original design targets a weaker threshold. On the other hand, running on plain Narwhal also removes the latency penalty that weak edges impose by waiting for slow replicas, so the performance numbers we report for FairDAG-RL in this paper can be read as an upper bound on what its fairness layer can achieve. Interestingly, by implementing FairDAG-RL from scratch given its protocol description, we have noticed a liveness bug which can be turned into a liveness attack as discussed in Appendix~\ref{liveness_attacks}.

When it comes to the Themis~\cite{Kelkar_Cornell_Themis_Nov2022} protocol, we have utilized the baseline implementation from the authors of Rashnu~\cite{Nagda_UPenn_Rashnu_May2024}~\footnote{https://github.com/HeenaNagda/Themis\_tx}, as the original released Themis source code is not complete. We have further adjusted the codebase by allowing for transaction size parameter control and have changed client broadcast to mimic the Narwhal client broadcast approach. 

\begin{figure*}[t]
    \centering
    \begin{subfigure}[t]{0.48\textwidth}
        \centering
        \includegraphics[width=\linewidth]{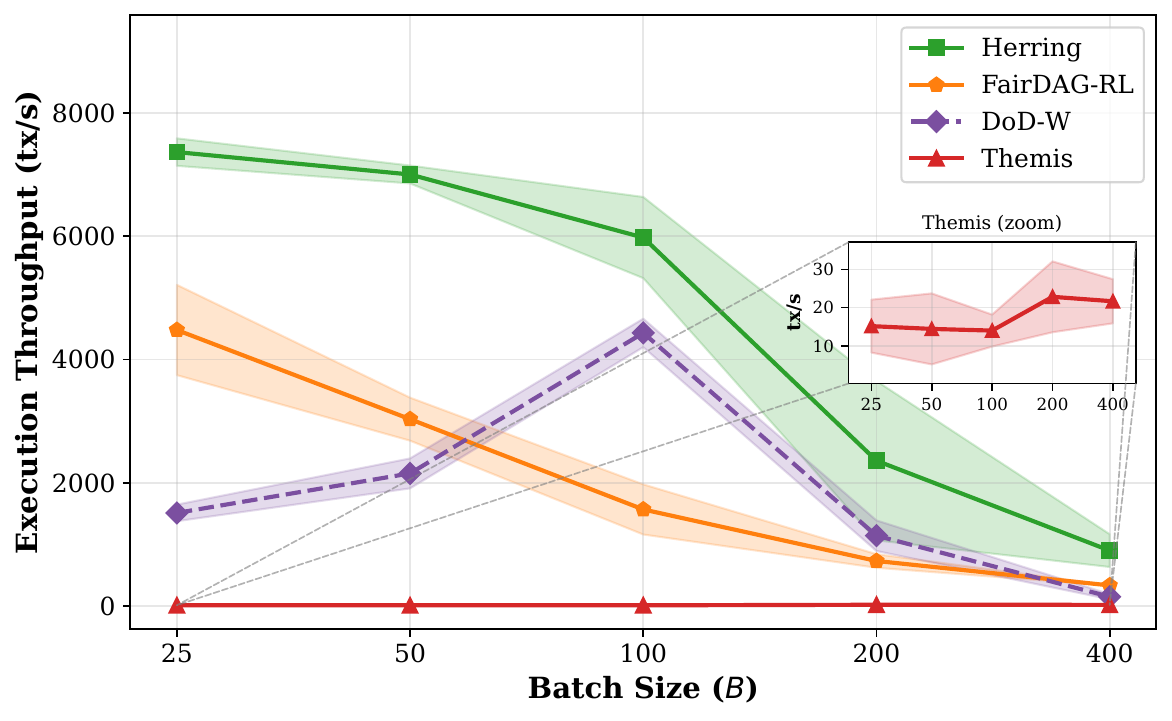}
        \label{fig:exp_5_2_3_throughput}
    \end{subfigure}\hfill
    \begin{subfigure}[t]{0.48\textwidth}
        \centering
        \includegraphics[width=\linewidth]{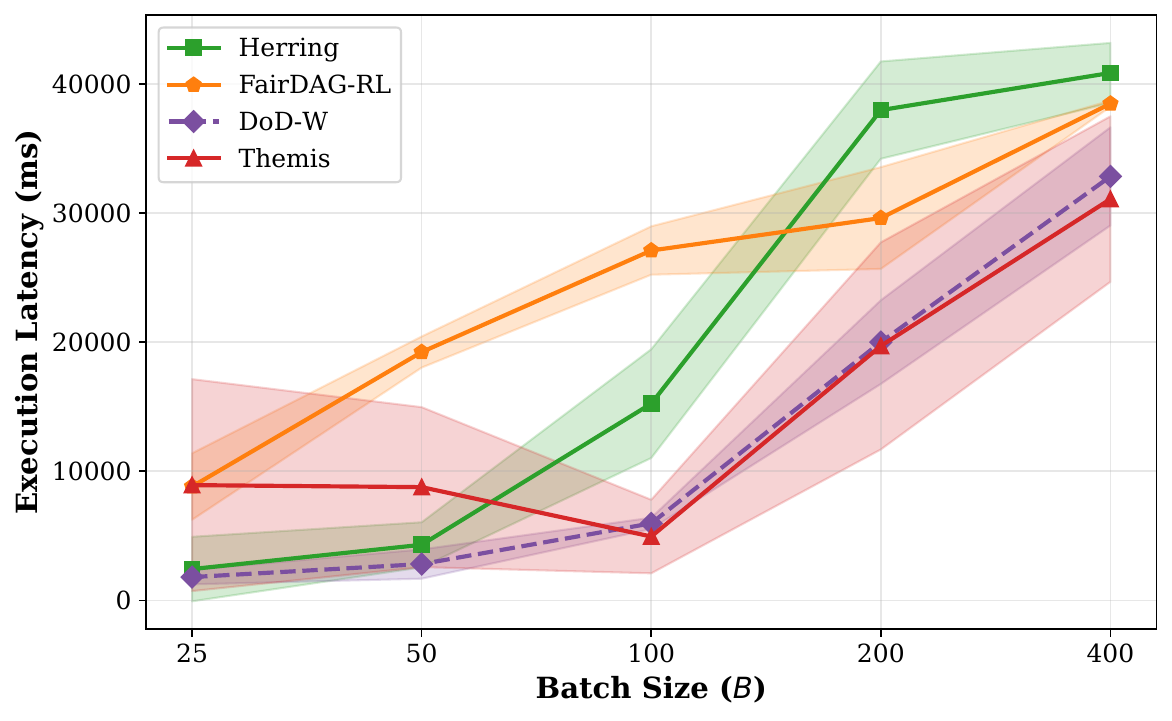}
        \label{fig:exp_5_2_3_latency}
    \end{subfigure}
    \caption{Impact of the batch size ($B$) parameter on performance of Herring, FairDAG-RL, DoD-W and Themis at an input rate of 7{,}000\,tx/s with $N=13$, $f=3$ and $\gamma = 1.0$.}
    \label{fig:exp_5_2_3}
\end{figure*}

\subsection{Setup}

We run our experiments on the DAS5~\cite{das5} academic cluster where each cluster node consists of two 8-core processors clocked at 2.4\,GHz and 64\,GB of RAM DDR3 memory. Nodes are interconnected via InfiniBand (IB) and Gigabit Ethernet (GbE). Each experiment run lasted for 60 seconds with each reported result being the average of at least 5 runs. Experiments mainly measure end-to-end (execution) throughput and latency of the system.

Each Narwhal node consists of one primary process together with one worker process and each Narwhal node was assigned a separate machine, i.e., primary and worker processes of the same node were collocated on the same machine. When it comes to the main Narwhal parameters, we have chosen a fixed size of $128$ bytes for each transaction, $4000$ bytes for each batch (around $32$ transactions per batch) and $512$ bytes for the block header size (around $16$ batches per block). 

We spawn $N$ clients, where $N$ is the number of nodes and each client connects to every node. Thus, every replica receives transactions from multiple clients in parallel, ensuring a saturated and uniformly distributed input load across the system. Each client submits transactions in bursts with a precision of $20$, meaning the target send rate is divided into $20$ equally-spaced bursts per second (one burst every $50$\,ms). To avoid duplicate transactions across clients, each client initializes its transaction counter to a random $64$-bit value and shuffles the order of its outgoing connections on every burst, so that no two clients submit overlapping transaction identifiers and no replica is favored. We assign a separate compute node per each group of four clients and the same client setup is used across all evaluated protocols to ensure a fair comparison.

\subsection{Different Batch Sizes}
We first evaluate the impact of the batch size parameter on the performance of Themis, DoD-W, FairDAG-RL and Herring. The number of replicas is set to $N=13$, the batch-OF parameter $\gamma$ is set to $\gamma=1.0$ and the input rate is set to 7{,}000\,tx/s. We set the number of faulty replicas $f$ to the maximum allowed given the constraints, thus we have $f=3$ replicas that silently crash. Similarly to previous batch-OF protocols, we compare the protocols across the following batch sizes: $\{25, 50, 100, 200, 400\}$. The batch size value represents the number of transaction orderings present within a single block.

\begin{figure*}[b]
    \centering
    \begin{subfigure}[t]{0.48\textwidth}
        \centering
        \includegraphics[width=\linewidth]{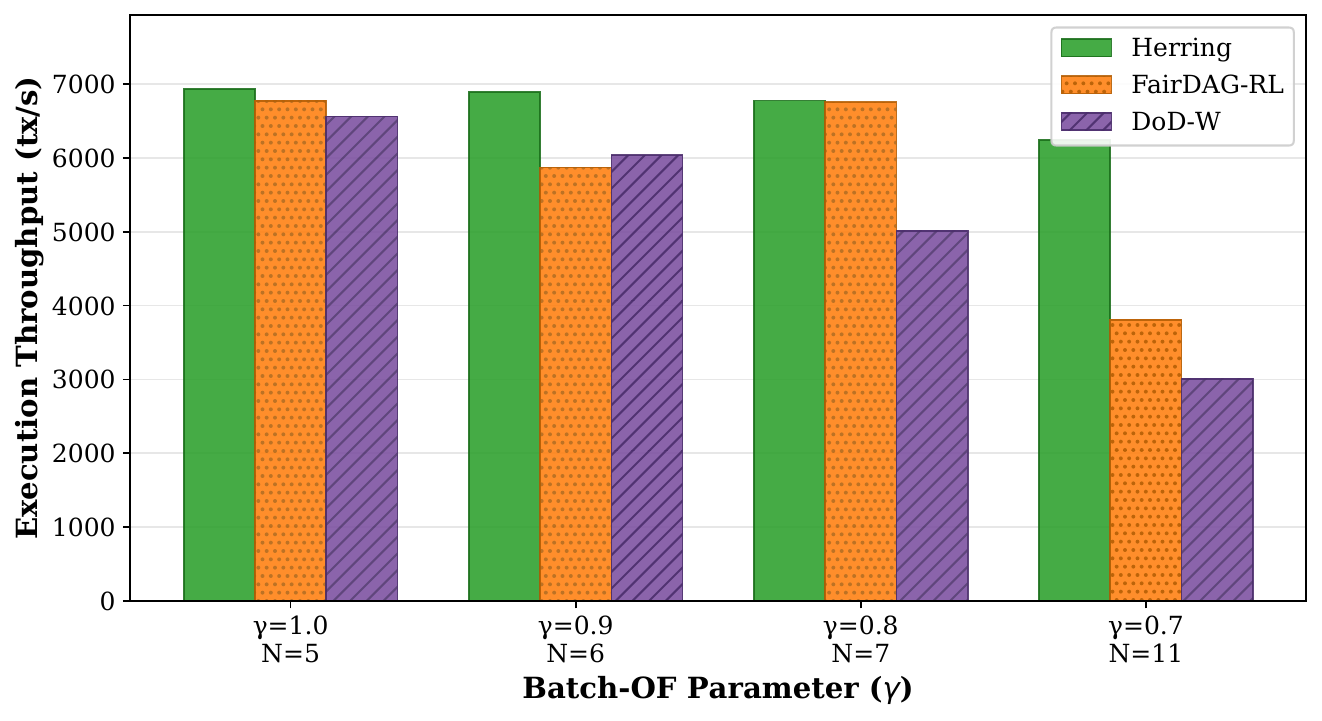}
        \label{fig:exp_5_2_4_throughput}
    \end{subfigure}\hfill
    \begin{subfigure}[t]{0.48\textwidth}
        \centering
        \includegraphics[width=\linewidth]{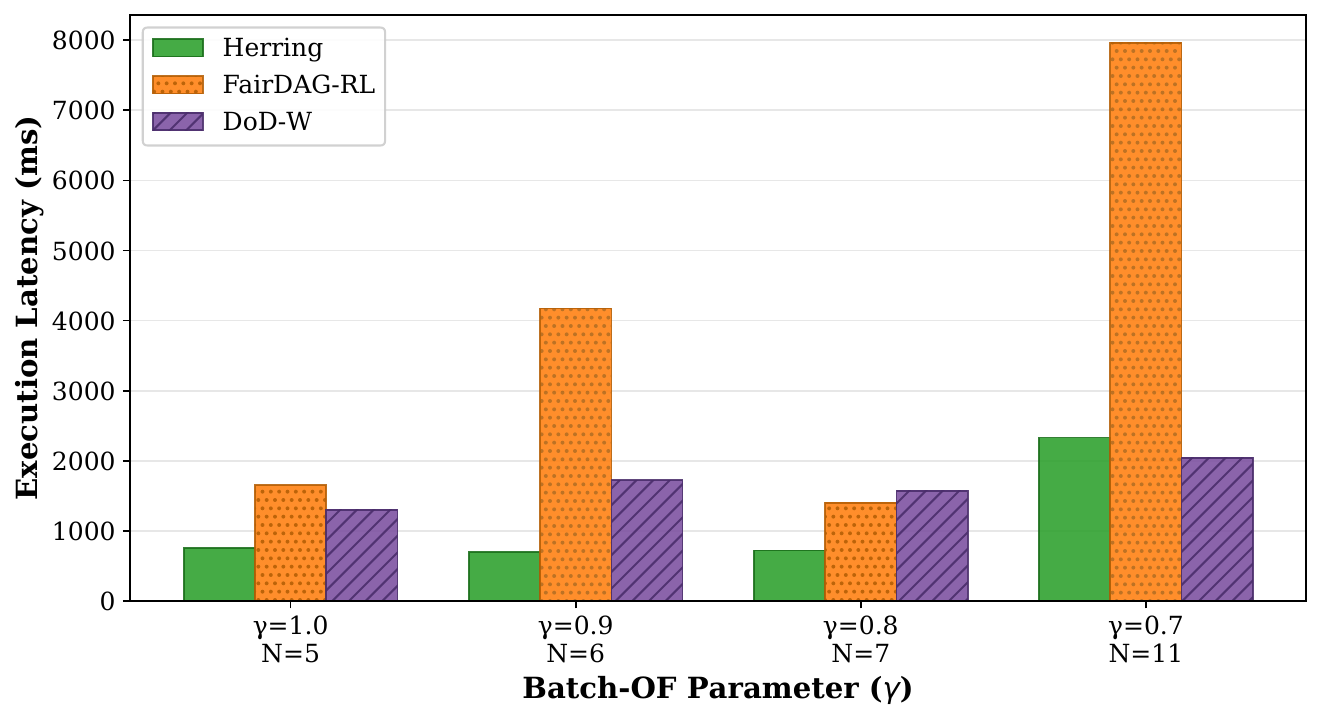}
        \label{fig:exp_5_2_4_latency}
    \end{subfigure}
    \caption{Impact of the batch-OF parameter $\gamma$ on performance of Herring, FairDAG-RL and DoD-W at an input rate of 7{,}000\,tx/s with $f=1$.}
    \label{fig:exp_5_2_4}
\end{figure*}

In leader-based protocols such as Themis, consensus occurs immediately after receiving the first block from each of the $N-f$ replicas, meaning the computational complexity of global graph construction is directly tied to the batch size value, i.e., a batch size of $B$ yields a global graph built from exactly $(N-f) \cdot B$ transaction orderings per consensus. In contrast, in DAG-based protocols such as DoD-W, FairDAG-RL and Herring, a single consensus instance can commit multiple blocks per replica at once and thus the effective number of transaction orderings fed into global graph construction is in addition bounded by the depth of the committed sub-DAG. As a result, for the same batch size value, DAG-based protocols generally process a larger number of transaction orderings per consensus instance than leader-based ones, which needs to be kept in mind when interpreting the results.

Figure~\ref{fig:exp_5_2_3} shows that Herring consistently achieves the highest execution throughput across all evaluated batch sizes, sustaining close to the offered load of $7{,}000$\,tx/s for $B \leq 50$ and degrading gracefully as $B$ grows. FairDAG-RL peaks at the smallest batch size with roughly $4{,}500$\,tx/s and degrades monotonically thereafter, while DoD-W exhibits a markedly different, non-monotonic profile, rising from about $1{,}500$\,tx/s at $B=25$ up to a peak of roughly $4{,}400$\,tx/s at $B=100$ before collapsing at larger batch sizes. This behavior stems from the DoD-W architecture, where a global order graph is constructed \emph{prior} to consensus at every DAG round, effectively adding an extra global-ordering step on top of the underlying DAG BFT protocol, i.e., at small $B$ the fixed per-round cost of this step dominates and caps throughput, while at $B=100$ it is better amortized across more transaction orderings, yielding the observed peak. Themis is effectively off-scale in comparison, sustaining only around $8$--$30$\,tx/s across the entire range (see inset), which reflects the fundamental throughput gap between leader-based and DAG-based batch-OF designs. The latency results show that Herring and DoD-W achieve the lowest execution latencies at small batch sizes (both around $2$\,s at $B=25$), whereas both FairDAG-RL and Themis already incur close to $10$\,s at the same batch size point. As $B$ grows, all protocols see their latency increase substantially which is consistent with our earlier observation that batch-OF protocols feed a larger number of transaction orderings into global graph construction per consensus instance as $B$ grows, thus increasing computation time. Overall, these results indicate that Herring offers the best throughput-latency trade-off in the small-to-moderate batch size regime ($B < 100$).

\subsection{Varying the Batch-Order-Fairness Parameter}
In this experiment, we analyze the impact of the batch-OF parameter $\gamma$ on performance of the fairness DAG BFT protocols. The network size $N$ and the batch-OF parameter $\gamma$ are correlated with the following formula: $N > \frac{4f}{2\gamma - 1}$, thus by decreasing $\gamma$ the network size $N$ increases. We make $f=1$ and we vary $\gamma$ to be: $\gamma=1.0\text{ }(N=5)$, $\gamma=0.9\text{ }(N=6)$, $\gamma=0.8\text{ }(N=7)$ and $\gamma=0.7\text{ }(N=11)$. For each protocol we select the batch size yielding the best performance based on the results of the previous experiment, thereby we have $B=25$ for Herring and FairDAG-RL and $B=100$ for DoD-W. The input rate is kept at $7{,}000$\,tx/s.

Figure~\ref{fig:exp_5_2_4} illustrates the throughput and latency across Herring, FairDAG-RL and DoD-W for different $\gamma$ values. At $\gamma \in \{1.0, 0.9, 0.8\}$, Herring consistently tracks the offered load of $7{,}000$\,tx/s, sustaining roughly $6{,}900$\,tx/s across all three settings at sub-second latency (around $700$--$750$\,ms). FairDAG-RL reaches comparable throughput at $\gamma=1.0$ and $\gamma=0.8$ (around $6{,}800$\,tx/s), but with significantly higher variance and a latency spike to roughly $4.2$\,s at $\gamma=0.9$. DoD-W remains close to the offered load at $\gamma=1.0$ but begins to fall behind at $\gamma=0.8$, dropping to $5{,}000$\,tx/s.

The gap opens clearly at $\gamma=0.7$ ($N=11$), where Herring sustains around $6{,}200$\,tx/s at $2.3$\,s of execution latency, while FairDAG-RL drops to $3{,}800$\,tx/s at $8.0$\,s and DoD-W drops to $3{,}000$\,tx/s at $2.0$\,s. This corresponds to a $64\%$ throughput improvement over FairDAG-RL and a $108\%$ improvement over DoD-W, together with $71\%$ lower latency than FairDAG-RL. DoD-W attains comparable latency to Herring at $\gamma=0.7$, but only at roughly half the throughput. This widening gap at stricter fairness settings highlights the scalability advantage of Herring's design, which absorbs the growth of $N$ substantially better than both baselines due to its intrinsic parallelism.

\subsection{Increasing Input Rates}
This experiment measures the performance of the protocols by increasing the transaction input rate into the system. The intention behind the experiment is to determine the saturation points for each protocol, i.e., points with maximum achievable input rate after which the protocol stagnates in processing transactions and latency starts to sharply increase. We choose $N=13$, $\gamma=1.0$ and set $f=3$ replicas to silently crash. In addition to the comparison between fairness DAG BFT protocols, we include the underlying DAG consensus protocol of Narwhal \& Tusk, together with the state-of-the-art batch-OF leader-based Themis protocol. Similarly to the previous experiment, per protocol we select the batch size yielding best performance. 

\begin{figure*}[b]
    \centering
    \begin{subfigure}[t]{0.48\textwidth}
        \centering
        \includegraphics[width=\linewidth]{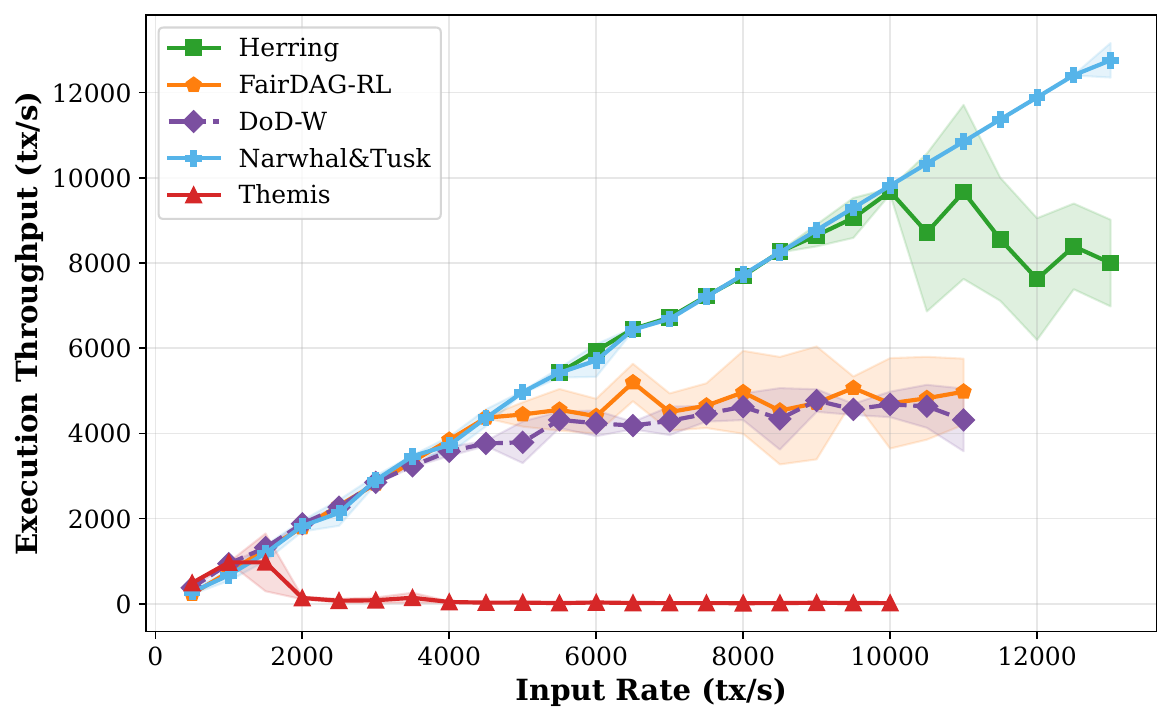}
        \label{fig:exp_5_3_1_throughput}
    \end{subfigure}\hfill
    \begin{subfigure}[t]{0.48\textwidth}
        \centering
        \includegraphics[width=\linewidth]{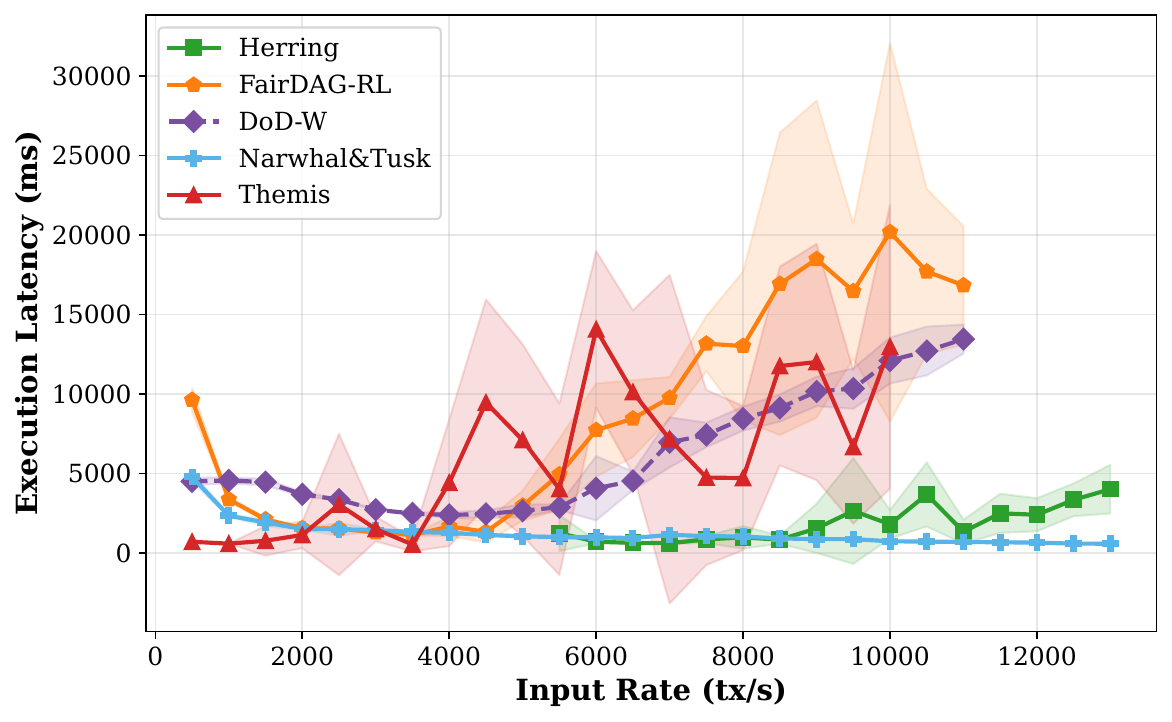}
        \label{fig:exp_5_3_1_latency}
    \end{subfigure}
    \caption{Impact of the client input rate on performance of Herring, FairDAG-RL, DoD-W, \textit{Narwhal\&Tusk} and Themis at with $N=13$, $f=3$ and $\gamma=1.0$.}
    \label{fig:exp_5_3_1}
\end{figure*}

\begin{figure*}[t]
    \centering
    \begin{subfigure}[t]{0.48\textwidth}
        \centering
        \includegraphics[width=\linewidth]{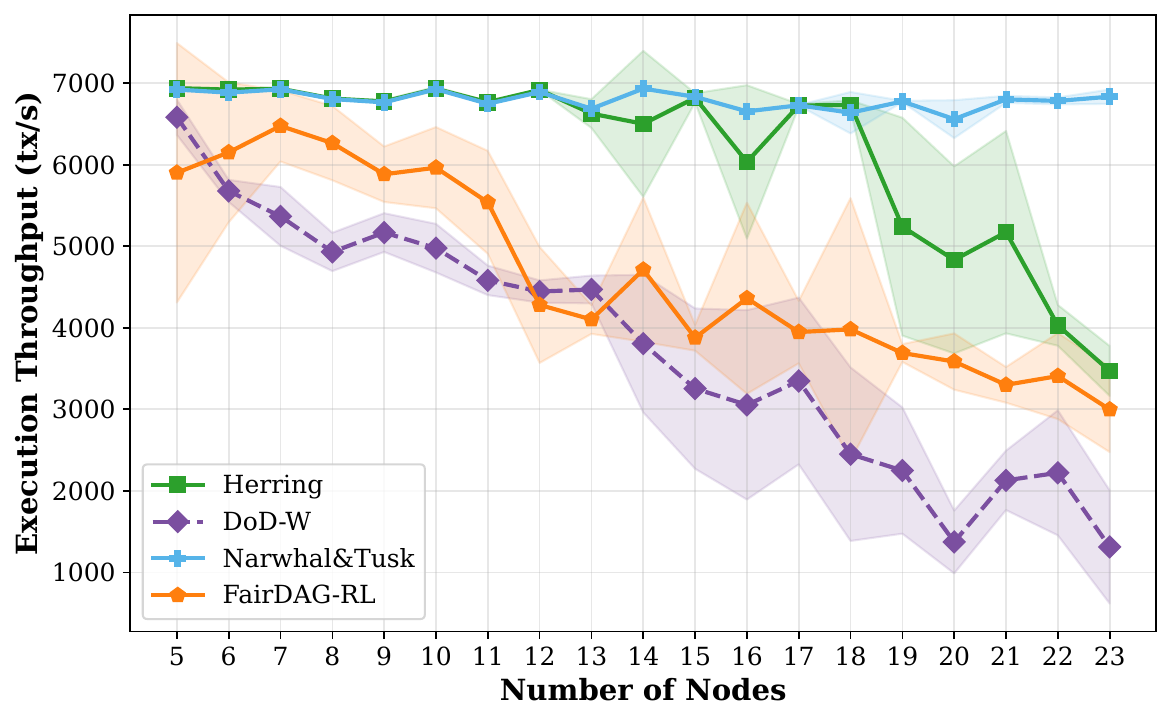}
        \label{fig:exp_5_3_2_throughput}
    \end{subfigure}\hfill
    \begin{subfigure}[t]{0.48\textwidth}
        \centering
        \includegraphics[width=\linewidth]{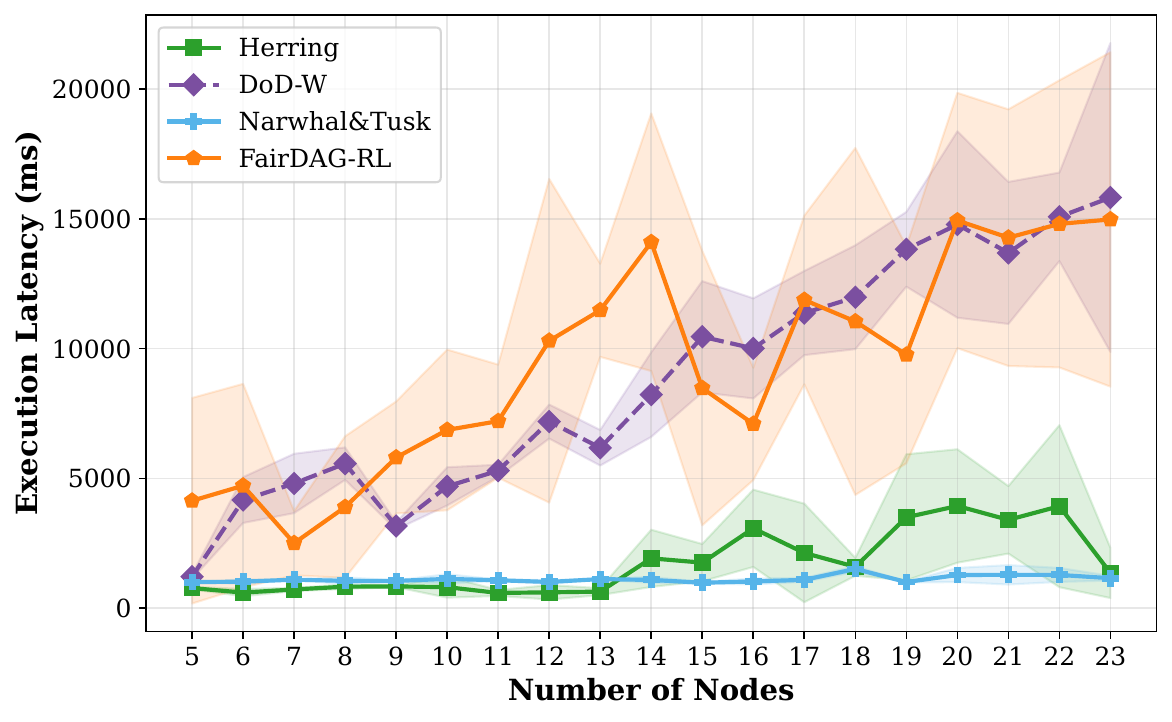}
        \label{fig:exp_5_3_2_latency}
    \end{subfigure}
    \caption{Impact of the network size $N$ on performance of Herring, FairDAG-RL, DoD-W and \textit{Narwhal\&Tusk} at an input rate of 7{,}000\,tx/s with $\gamma=1.0$.}
    \label{fig:exp_5_3_2}
\end{figure*}

Figure~\ref{fig:exp_5_3_1} illustrates the throughput and latency of all five protocols as the input rate is increased from $500$\,tx/s up to $13{,}000$\,tx/s. Most notably, Herring tracks the throughput curve of the underlying Narwhal \& Tusk protocol up to roughly $10{,}000$\,tx/s, at which point it peaks at around $9{,}600$\,tx/s before entering a noisy regime with a gradual decline. This demonstrates that the parallel design nature of the fair ordering layer of Herring introduces minimal overhead on top of its underlying DAG consensus. In contrast, both FairDAG-RL and DoD-W saturate much earlier, at around $5{,}000$\,tx/s and $4{,}500$\,tx/s respectively, after which their throughput flattens and their execution latency climbs sharply, reaching close to $18$--$20$\,s for FairDAG-RL and $13$\,s for DoD-W at the highest offered rates. Themis saturates the earliest, already around $1{,}500$\,tx/s, after which its throughput collapses to effectively zero and its latency starts increasing, confirming the fundamental scalability gap between leader-based and DAG-based batch-OF designs. When looking at saturation points, Herring achieves roughly $90\%$ higher throughput than FairDAG-RL, roughly $100\%$ higher throughput than DoD-W, and more than an order of magnitude higher throughput than Themis, while simultaneously sustaining the lowest execution latency among all fairness protocols across the entire evaluated range. These results confirm that Herring is the only evaluated fair ordering protocol that is able to keep up with the performance of its underlying DAG consensus protocol, making fair ordering essentially a near-free property in its design.

\subsection{Different Network Sizes}
In this set of experiments we measure the performance of the DAG BFT protocols by varying the network size, i.e., we keep the input rate fixed at $7{,}000$\,tx/s and vary the number of nodes from 5 until 23, thereby also adjusting the fault threshold $f$ from 1 to 5. Faulty nodes are silently crashing. We keep the batch-OF parameter $\gamma$ as $\gamma=1.0$ and again use the optimal batch size per each protocol. We exclude Themis from this set of experiments as its performance compared to DAG BFT protocols has been evaluated sufficiently, i.e., it is expected that it performs much worse than other protocols.

\begin{figure*}[b]
    \centering
    \begin{subfigure}[t]{0.48\textwidth}
        \centering
        \includegraphics[width=\linewidth]{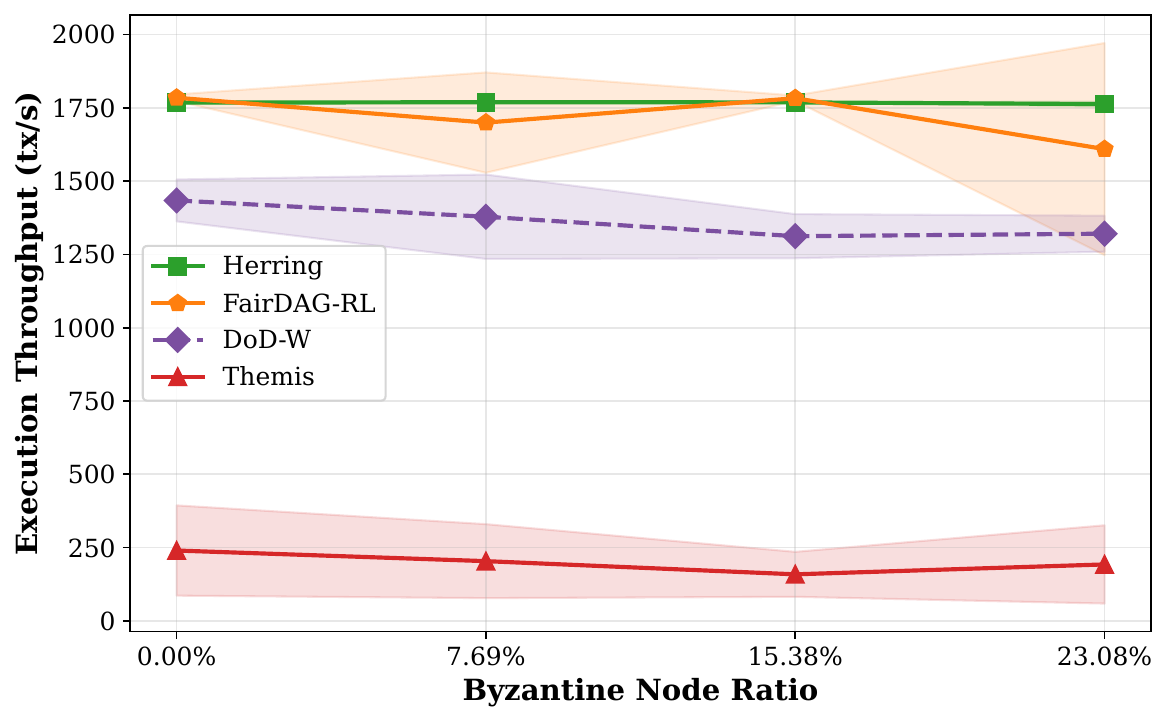}
        \label{fig:exp_faulty_throughput}
    \end{subfigure}\hfill
    \begin{subfigure}[t]{0.48\textwidth}
        \centering
        \includegraphics[width=\linewidth]{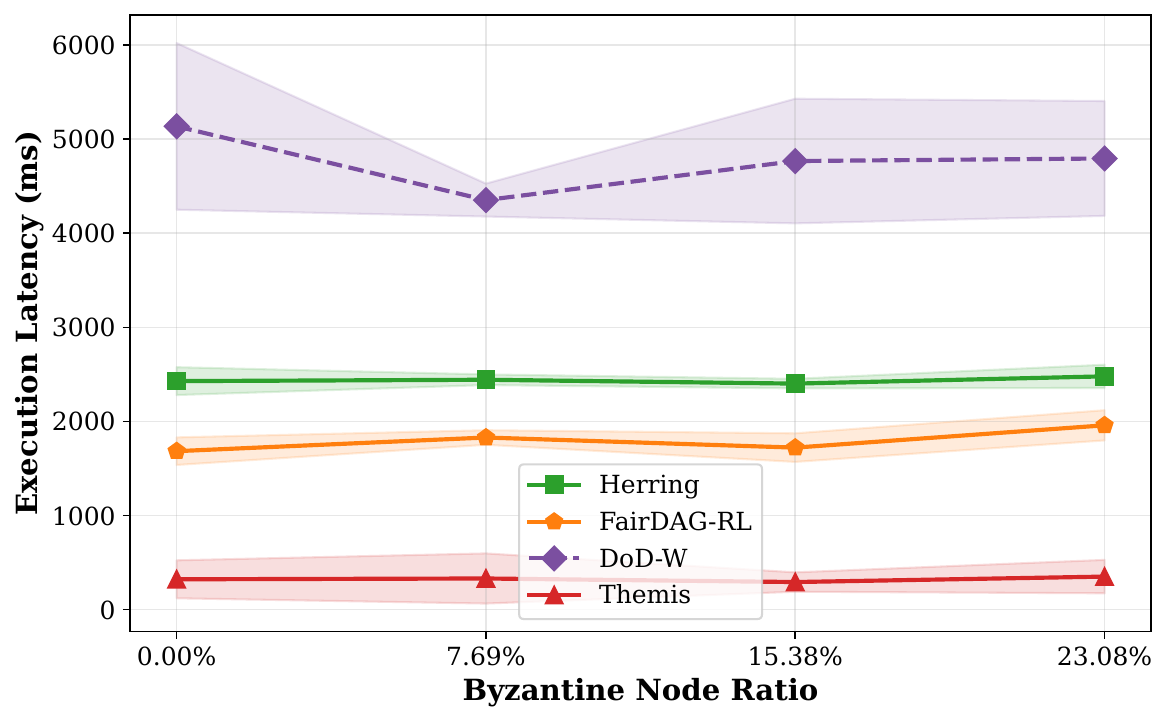}
        \label{fig:exp_faulty_latency}
    \end{subfigure}
    \caption{Impact of different ratios of Byzantine nodes on performance of Herring, FairDAG-RL, DoD-W and Themis at an input rate of 2{,}000\,tx/s with $N=13$ and $\gamma=1.0$.}
    \label{fig:exp_faulty}
\end{figure*}

Figure~\ref{fig:exp_5_3_2} shows how the throughput and latency of fairness DAG BFT protocols degrade with the increase in the number of nodes, in contrast to the underlying DAG BFT protocol of Narwhal \& Tusk, which sustains a flat throughput and latency. This is expected due to the increase in communication overhead and computational costs, as each global graph construction computational cost is correlated with the amount of local receive orderings ($N-f$ factor). Most notably, Herring closely matches the throughput of Narwhal \& Tusk up to $N=18$, sustaining close to the offered $7{,}000$\,tx/s, and only starts to visibly degrade beyond $N=18$, eventually settling at around $3{,}500$\,tx/s for $N=23$. Its latency remains within a few seconds across the entire evaluated range, peaking at around $4$\,s at $N=20$ before dropping back. In contrast, both FairDAG-RL and DoD-W degrade from the start, with FairDAG-RL dropping from around $5{,}900$\,tx/s at $N=5$ to roughly $3{,}000$\,tx/s at $N=23$, while DoD-W degrades even more aggressively, falling from $6{,}500$\,tx/s down to around $1{,}300$\,tx/s over the same range. Their latencies grow correspondingly, reaching $15$\,s and $16$\,s respectively at $N=23$, i.e., roughly an order of magnitude higher than Herring at the same network size. At $N=18$, Herring achieves around $70\%$ higher throughput than FairDAG-RL and over $175\%$ higher throughput than DoD-W, while sustaining roughly $85\%$ lower latency than both. These results confirm that Herring scales substantially better with network size than the existing fairness DAG BFT protocols.

\subsection{Different Ratios of Byzantine Nodes}
This experiment measures the performance of all fairness DAG BFT protocols together with the leader-based protocol of Themis under different ratios of Byzantine nodes. Byzantine nodes perform a simple yet effective attack where they reverse the local ordering of transactions. We fix the network size at $N=13$, keep the batch-OF parameter $\gamma$ at $\gamma=1.0$ and use the optimal batch size per protocol. We set the unique transaction input rate to the system to be $2{,}000$\,tx/s, as this is the point before the saturation of Themis, which thus allows us to fairly compare all protocols. We vary the ratio of Byzantine nodes as $f=0\text{ }(0\%)$, $f=1\text{ }(7.69\%)$, $f=2\text{ }(15.38\%)$ and $f=3\text{ }(23.08\%)$.

Figure~\ref{fig:exp_faulty} shows the throughput and latency of all four protocols under varying Byzantine ratios. Overall, all protocols remain relatively stable across the evaluated range, with only minor fluctuations in both metrics, which confirms that the attack primarily affects the ordering of transactions rather than raw performance. Herring and FairDAG-RL deliver the highest throughput, both sustaining close to $1{,}750$\,tx/s across all ratios, followed by DoD-W at around $1{,}300$--$1{,}450$\,tx/s, while Themis remains effectively saturated at around $200$\,tx/s. Among the DAG-based protocols, Herring exhibits the smallest throughput reduction as the Byzantine ratio grows, dropping by less than $1\%$ between $0\%$ and $23.08\%$, compared to roughly $9\%$ for FairDAG-RL and $8\%$ for DoD-W. Interestingly, FairDAG-RL achieves a slightly lower execution latency than Herring at this operating point. This is because at $2{,}000$\,tx/s neither protocol is close to its saturation point, and FairDAG-RL's implicit edge update mechanism avoids the extra communication overhead that Herring's explicit edge resolution incurs, while Herring's parallel graph construction advantage only materializes at higher input rates, as shown in previous experiments. Themis is an outlier in terms of latency, sustaining around $300$\,ms as it operates below its saturation point at the fixed input rate, but at the cost of an order of magnitude lower throughput than the DAG-based protocols. These results confirm that Herring remains robust under Byzantine behavior, with its throughput essentially unaffected.

\subsection{Reversing Order Attack Resistance}
In this experiment, we evaluate the robustness of Herring against adversarial reordering compared to Themis, following the methodology of Kelkar et al.~\cite{Kelkar_Cornell_Themis_Nov2022}. We choose Themis as the comparison baseline, as it is the state-of-the-art leader-based batch-OF protocol, and our goal here is to highlight the robustness advantage of DAG-based batch-OF protocols due to their underlying data dissemination design. Each Byzantine replica reverses its local ordering before reporting it. As a closeness measure between two transactions we use $\text{Dist}(\textit{tx}, \textit{tx}') = |\#(\textit{tx} \prec \textit{tx}') - \#(\textit{tx}' \prec \textit{tx})|$, where small $\text{Dist}$ indicates a fragile pair. Concretely, $\text{Dist}$ ranges from $N \bmod 2$ up to $N$ in steps of $2$, where $\text{Dist}=N$ means that all nodes agree on the receive order of the pair, and for each $\text{Dist}$ bucket the y-axis reports the fraction of transaction pairs whose final total ordering disagrees with the honest majority receive order. We fix $N=21$, $f=5$ and vary $f_{\text{actual}} \in \{0, 2, 5\}$.

\begin{figure}[H]
    \centering
    \includegraphics[width=\linewidth]{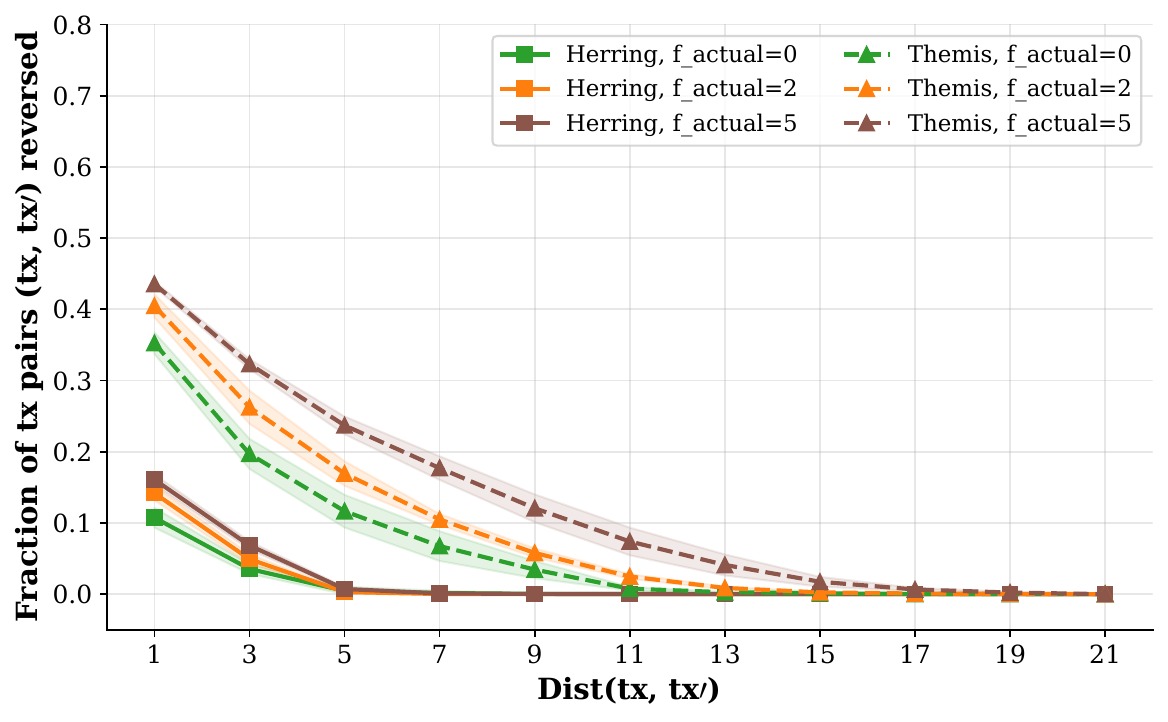}
    \caption{Robustness against the reversing order adversarial attack for Herring and Themis with $N=21$, $f=5$, and varying the actual number of Byzantine replicas $f_{\text{actual}} \in \{0, 2, 5\}$.}
    \label{fig:exp_adv_reorder}
\end{figure}

Figure~\ref{fig:exp_adv_reorder} shows that Herring consistently outperforms Themis in resilience across all settings. At $\text{Dist}=1$, Herring reverses only $\sim\!11\%$--$16\%$ of pairs across all $f_{\text{actual}}$ values, whereas Themis reverses $\sim\!35\%$--$44\%$. More importantly, the fraction of reversed pairs in Herring drops to effectively zero already at $\text{Dist}=5$ regardless of $f_{\text{actual}}$, while Themis only converges to zero at $\text{Dist} \geq 15$ and exhibits a long tail in between. The gap between the three $f_{\text{actual}}$ curves is also markedly smaller for Herring, indicating that the additional leverage an adversary gains from controlling more replicas is much more limited. Similar to FairDAG~\cite{Kang_FairDAG_Aug2025} results, the results substantiate the claim that DAG-based batch-OF protocols structurally mitigate adversarial ordering manipulation due to their inherent reliable dissemination layer.

\section{Related Work}

\noindent\textbf{Order-fairness protocols.}
The notion of $\gamma$-batch-order-fairness was introduced by Kelkar et al.~\cite{Kelkar_Cornell_Aequitas_Aug2020} with Aequitas, which however only provides weak-liveness due to arbitrarily chained Condorcet cycles.
Themis~\cite{Kelkar_Cornell_Themis_Nov2022} restored standard liveness through batch unspooling and deferred ordering and serves as the direct ancestor of the leader-based fair-ordering line, which subsequently evolved along axes of data-dependent fairness~\cite{Nagda_UPenn_Rashnu_May2024}, pipelining with consensus~\cite{Mu_China_SpeedyFair_Jan2024}, weight-based sorting~\cite{Chen_China_Auncel_Aug2024}, position fairness~\cite{Wang_China_Dikaios_Oct2025}, cross-round graph maintenance~\cite{Ren_Melbourne_AUTIG_Oct2025}, bounded unfairness~\cite{Kiayias_Edinburgh_Taxis_2024}, and minimal batch variants~\cite{Ramseyer_Stanford_StreamingSocial_Oct2024}. Such work is orthogonal to Herring and we leave its evaluation as future work.
Some approaches replace local orderings with median or inferred timestamps~\cite{Zhang_Cornell_Pompe_Nov2020,Kursawe_Wendy_Jul2020,Zarbafian_Sydney_Lyra_May2023}, hide transaction contents until commit~\cite{Stathakopoulou_TEEs_2021,Malkhi_Chainlink_Dec2022}, or randomly permute committed sets~\cite{Yakira_Helix_2021,Alpos_EatingSandwiches_2023,Kavousi_Blindperm_2023}.

\smallskip
\noindent\textbf{Attacks on batch-OF.}
The Condorcet~\cite{Vafadar_Alberta_Condorcet_Jun2023} and Ambush~\cite{Park_KAIST_Ambush_Jun2025} attacks exploit client-induced information asymmetry across nodes to force Condorcet cycles that capture honest transactions, and both works independently propose internal transaction gossiping among nodes as the mitigation.
This mitigation is provided natively by Narwhal's dissemination layer, where every worker batch is consistently broadcast to the corresponding workers on all other nodes, and thus transfers to Herring without additional machinery. Figure~\ref{fig:exp_adv_reorder} possibly hints at DAG BFT protocols being resistant by design to such batch-OF attacks. 

\smallskip
\noindent\textbf{DAG BFT consensus.}
A parallel line of work has re-architected BFT consensus around DAGs in order to remove the single-leader bottleneck.
\textit{Certified} designs enforce non-equivocation at the dissemination layer through reliable broadcast~\cite{DAG_Rider_Jun2021,Narwhal&Tusk_Mar2022,Bullshark_Sep2022,Shoal_2023,Shoal++_2025}.
\textit{Uncertified} designs relax this to best-effort broadcast in exchange for lower commit latency and handle equivocation at the ordering layer~\cite{Cordial_Miners_2023,Mysticeti_2025,Mahi-Mahi_2025,Starfish_2025}.
\textit{Hybrid} designs use the DAG only as a dissemination layer and defer ordering to a leader-based finalization protocol~\cite{Autobahn_2025}.
None of these protocols enforce any OF property. As Herring works post-consensus and requires minimal changes in the underlying DAG (addition of LOIs and the self-referencing rule), we believe Herring's design can be easily applied to any such DAG BFT in order to enable the OF property. We leave this as future work.

\smallskip
\noindent\textbf{Fair ordering on DAG BFT.}
Only FairDAG-RL~\cite{Kang_FairDAG_Aug2025} and DoD~\cite{Nagda_UPenn_PhD_2025} realize batch-OF on top of DAG BFT, and both are dissected in detail against Herring in Section~\ref{sec:herring}, Appendix~\ref{liveness_attacks}, and the experimental evaluation.
On the consensus-layer side, Zhang et al.~\cite{Zhang_Purdue_Oct2024} demonstrate inter-block frontrunning attacks targeting the DR-first ordering rule of \textit{Tusk} and \textit{Bullshark}, which are orthogonal to the batch-OF layer and apply to Herring only to the extent that they apply to its underlying DAG BFT consensus. Upon playing with the released code of Zhang et al.~\cite{Zhang_Purdue_Oct2024} we noticed that the attack activation code does not reflect the theoretical definition given in the paper, and thus we defer the proper evaluation as future work.

\smallskip
\noindent\textbf{Liveness of prior batch-OF DAG BFT protocols.}
While reimplementing FairDAG-RL~\cite{Kang_FairDAG_Aug2025} and DoD~\cite{Nagda_UPenn_PhD_2025} from their pseudocode descriptions as baselines, we identified previously unreported liveness bugs in the fairness layer of each protocol. In FairDAG-RL, the weight-update loop of the graph construction algorithm only considers ordering indicators that arrive in the current subdag, silently discarding indicators deposited while the transaction was still blank. A concrete input sequence can make the dependency graph permanently fail to become a tournament, stalling all future finalization. In DoD, three independent issues interact to prevent the missing edge weight store from ever reaching the edge threshold for a common class of transaction pairs, stalling the execution queue. We describe the bugs, concrete triggering scenarios, and the patches we integrated into our reimplementations in Appendix~\ref{liveness_attacks}.

\section{Conclusion}

We presented Herring, the first batch-OF DAG BFT protocol whose dominant fairness layer cost runs embarrassingly parallel across committed subdags.
The design starts from the observation that the dominant cost in any Themis-style fairness layer, the pairwise weight matrix, has no cross subdag data dependency, and exploits this by running multiple subdags' graph construction tasks concurrently on a thread pool.
Two supporting design choices keep this parallelism sound: post-consensus graph construction (shared with FairDAG-RL) that keeps all fairness work off the DAG's critical path, and explicit FairUpdate vote resolution (inspired by Themis) that piggybacks on Narwhal's batch dissemination and lets each parked graph accumulate its resolutions independently.
Along the way, we identified and patched previously unreported liveness bugs in both FairDAG-RL and DoD.
Our evaluation shows that Herring tracks the throughput of Narwhal \& Tusk closely up to roughly $10{,}000$\,tx/s, sustains up to $100\%$ higher throughput and $75\%$ lower latency than the two DAG-based baselines, and remains essentially unaffected by colluding Byzantine nodes.

\printbibliography

\section{Open Science}

This paper presents Herring, a DAG-based consensus protocol that is designed to ensure batch-order fairness without sacrificing high performance.
To support reproducibility and independent verification, we have open-sourced the complete implementation of Herring, alongside our implementations of the FairDAG-RL and DoD-W baselines, at \url{https://github.com/randomUserGithub123/narwhal}.

\section{Ethical Considerations}
This work raises no ethical concerns.
All data, source code, and measurement results used in this paper are publicly available and contain no personally identifiable information.
No human or animal subjects were involved in any part of the experimental procedure, and the work does not raise concerns related to health care, environmental impact, or military applications.

\appendix

\section{Full Correctness Proofs for Herring}
\label{sec:full_proofs}

Throughout this appendix, $\tau = n(1 - \gamma) + f + 1$ and $\tau_s = n - 2f$, with $n > 4f/(2\gamma - 1)$ and $\tfrac{1}{2} < \gamma \leq 1$. We use the cumulative state notation of Section~\ref{sec:herring_graph}, with an explicit round superscript where the proofs require it. $\Pi_i^r$ is replica $i$'s LOI ordered pending list at the moment $A_r$'s snapshot is extracted, i.e., the transactions that replica $i$ has locally observed and not yet contributed to any earlier $G_{r'}$. We treat $\Pi_i^r$ as a sequence when LOI order matters and as a set otherwise, and we drop the superscript and write $\Pi_i$ whenever the round is clear from context. $P^r$ is the set of transactions that have entered some $G_{r'}$ with $r' < r$, and $P_i^r$ is the subset that $i$ has locally assigned an LOI to, which we refer to as $i$'s graphed contributions up to $A_r$. $\{L_i^r\}$ is the snapshot extracted at $A_r$, where each $L_i^r$ is the LOI ordered prefix of $\Pi_i^r$ consisting of $i$'s contributions that have been sealed into DAG vertices committed by the time $A_r$ is processed. $K_r$ is the set of vertices retained in $G_r$ after anchor truncation.

The proof structure mirrors Themis~\cite{Kelkar_Cornell_Themis_Nov2022}. The first lemma shows that Herring's parallel execution is observationally equivalent to a serial reference in which each subdag's task runs to completion before the next begins. After establishing this equivalence, we argue about $G_r$ in isolation, exactly as Themis argues about a single leader proposal. The remaining safety and fairness lemmas lift from Themis with the substitution of the leader's collected list $L$ by the snapshot $\{L_i^r\}$ and of the leader proposal by $A_r$'s dependency graph $G_r$.

\begin{lemma}[Reduction to serial execution]
\label{lem:reduction}
Define the serial reference of Herring as the variant in which each $A_r$'s graph construction task runs to completion, including any subsequent ApplyFairUpdate resolution, before $A_{r+1}$'s task begins. For every committed sequence of subdags $A_1, A_2, \dots$, Herring's parallel execution and the serial reference emit the same total transaction order.
\end{lemma}

\begin{proof}
By induction on $r$, both $G_r$ and the set of FairUpdate edges admitted into $G_r$ agree in the parallel and serial executions.

For the base case, $A_1$'s task receives the empty cumulative chain in both executions and its snapshot $\{L_i^1\}$ is the same in both, since no prior exclusions exist. The four phases are deterministic functions of the snapshot, so $G_1$ and $K_1$ agree.

For the inductive step, assume the claim holds for all $r' < r$. The cumulative chain $\bigcup_{r' < r} K_{r'}$ is then identical in both executions. In the parallel execution, $A_r$'s snapshot at the synchronous extract step may additionally contain shaded vertices belonging to in flight predecessors, since those tasks have not yet returned their results to the main thread. The cumulative chain filter in Phase~2 removes exactly those vertices from the active set before edge construction. For every pair $(u, v)$ with both endpoints in the post filter active set, the orderings that contain both $u$ and $v$ are the same in the two executions, because the snapshots agree on every replica's contribution involving such pairs. The cached weights between active pairs therefore agree, and Phase~3 produces the same $G_r$ and $K_r$.

FairUpdate votes ride on Narwhal's reliable broadcast and are observed by both executions as part of the same committed vertex sequence. ApplyFairUpdate is a deterministic function of $(G_r, \text{votes})$, so the admitted FairUpdate edges agree.

Emit drains finalized orders in subdag commit order in both executions. Since the per subdag orders agree and the commit order is identical, the emitted total order agrees.
\end{proof}

A direct consequence is the single graph property, which Themis obtains for free from serial proposal construction.

\begin{corollary}[Single graph]
\label{cor:single_graph}
Each transaction enters at most one $G_r$.
\end{corollary}

\begin{proof}
In the serial reference each transaction is either retained in some $K_r$, and thereafter excluded from all snapshots $r' > r$ via the cumulative chain, or returned to pending state and eligible for the next snapshot. Lemma~\ref{lem:reduction} transfers this to the parallel execution.
\end{proof}

By Lemma~\ref{lem:reduction}, the parallel execution and the serial reference emit identical total orders on every committed sequence, so a property of the emitted order holds in one if and only if it holds in the other. For the remainder of the appendix we argue about the serial reference, and every conclusion applies to the parallel execution by equivalence.

\begin{lemma}[LOI monotonicity]
\label{lem:loi_monotone}
For each correct replica $i$ and each committed subdag $A_r$, the concatenation of $i$'s graphed contributions $P_i^r$ and its pending list $\Pi_i^r$, taken in LOI order, is non decreasing in LOI.
\end{lemma}
 
\begin{proof}
The self referencing rule of Section~\ref{sec:narwhal_tusk} forces each round $r$ vertex of correct replica $i$ to reference $i$'s round $r{-}1$ certificate. By induction on round, if $i$'s round $r$ vertex commits in $A_k$ then every earlier vertex of $i$ commits in some $A_j$ with $j \leq k$. Workers assign LOIs strictly monotonically on first observation and seal batches in LOI order, so transactions in $i$'s round $r$ vertex have LOIs above those in any earlier vertex of $i$, and walking through $i$'s committed vertices in subdag commit order therefore yields $P_i^r$ in LOI order. The Ingest step appends to $\Pi_i^r$ in LOI order, and removing the prefix retained in some earlier $G_{r'}$ preserves monotonicity of the remainder. Finally, every transaction in $P_i^r$ was assigned its LOI by $i$ in a round strictly earlier than any transaction still pending at $A_r$, so its LOI is below every entry of $\Pi_i^r$.
\end{proof}

\begin{lemma}[Graph structure and directional safety]
\label{lem:graph_directional}
The output of Algorithm~\ref{alg:herring_graph_new}, together with any subsequent ApplyFairUpdate resolution, satisfies the following.
\begin{enumerate}
\item Every solid transaction in $V_r$ lies in $K_r$.
\item For every solid $\mathit{tx}$ and non-blank $\mathit{tx}'$ in $V_r$, exactly one of $(\mathit{tx}, \mathit{tx}')$ and $(\mathit{tx}', \mathit{tx})$ is present in the finalized $G_r$.
\item If $\gamma n$ replicas received $\mathit{tx}$ before $\mathit{tx}'$, the edge $(\mathit{tx}', \mathit{tx})$ is never added to any $G_r$ nor installed by ApplyFairUpdate.
\end{enumerate}
\end{lemma}

\begin{proof}
Parts (1) and (2) follow the argument of Themis Lemma~B.1~\cite{Kelkar_Cornell_Themis_Nov2022} with $L$ replaced by $\{L_i^r\}$ and the leader proposal replaced by $G_r$. The threshold condition $n > 4f/(2\gamma - 1)$ ensures the larger direction reaches $\tau$ between any solid and non-blank pair, possibly through ApplyFairUpdate when $G_r$ is parked. Missing edges between two shaded vertices are resolved by ApplyFairUpdate using the same thresholds applied to the FairUpdate vote tallies. For part (3), at least $\gamma n - f$ correct replicas received $\mathit{tx}$ before $\mathit{tx}'$, so the wrong direction count in any snapshot or FairUpdate tally is at most $n(1 - \gamma) + f$, which is strictly less than $\tau$, and the threshold is never crossed in the wrong direction.
\end{proof}

\begin{lemma}[Solid precedes unobserved]
\label{lem:solid_precedes}
Suppose $\mathit{tx}$ is solid in $A_r$'s snapshot and $\mathit{tx}'$ has support strictly less than $\tau$ in $A_r$'s snapshot. Then at least $n(1 - \gamma) + 1$ correct replicas received $\mathit{tx}$ before $\mathit{tx}'$ in receive order.
\end{lemma}

\begin{proof}
The argument follows Themis Lemma~B.2~\cite{Kelkar_Cornell_Themis_Nov2022}. Since $\mathit{tx}$ is solid and $\mathit{tx}'$ has support below $\tau$, the number of snapshot orderings $L_i^r$ that contain $\mathit{tx}$ but not $\mathit{tx}'$ is at least $\gamma n - 3f$. At most $f$ of these come from Byzantine replicas, leaving at least $\gamma n - 4f$ correct replicas with $\mathit{tx} \in L_i^r$ and $\mathit{tx}' \notin L_i^r$ at the moment of $A_r$'s snapshot. For such a correct replica $i$, either $i$ has not yet observed $\mathit{tx}'$, in which case $i$ trivially received $\mathit{tx}$ first, or $i$ has observed $\mathit{tx}'$ but with an LOI past the prefix it contributed to $A_r$, so $\mathit{tx}' \in \Pi_i^r$. In the latter case Lemma~\ref{lem:loi_monotone} gives $\mathrm{LOI}_i(\mathit{tx}) < \mathrm{LOI}_i(\mathit{tx}')$, since $\mathit{tx}$ lies in $i$'s committed prefix $L_i^r$ and $\mathit{tx}'$ in its pending suffix $\Pi_i^r$, and the worker's monotonic LOI assignment forces $i$ to have received $\mathit{tx}$ first. Under $n > 4f/(2\gamma - 1)$, $\gamma n - 4f \geq n(1 - \gamma) + 1$ by integrality.
\end{proof}

\begin{lemma}[Contiguous cycles]
\label{lem:contiguous_cycles}
Suppose $\mathit{tx}$ finalizes in $G_s$ and $\mathit{tx}'$ in $G_r$ with $s > r$, and $\gamma n$ replicas received $\mathit{tx}$ before $\mathit{tx}'$. Then $\mathit{tx}$ and $\mathit{tx}'$ belong to the same Condorcet cycle.
\end{lemma}

\begin{proof}
The argument is the Themis Lemma~B.3~\cite{Kelkar_Cornell_Themis_Nov2022} case analysis applied to the serial reference. By Lemma~\ref{lem:graph_directional}(3) the edge $(\mathit{tx}', \mathit{tx})$ is never added.

If $\mathit{tx}$ has support at least $\tau$ in $A_r$'s snapshot, then $\mathit{tx} \in V_r$. At least $\gamma n - f \geq \tau$ orderings place $\mathit{tx}$ before $\mathit{tx}'$, so $(\mathit{tx}, \mathit{tx}') \in G_r$. In the condensation, $[\mathit{tx}]$ either coincides with $[\mathit{tx}']$ or precedes it topologically, and since $\mathit{tx}' \in K_r$ both cases give $\mathit{tx} \in K_r$, contradicting $s > r$ by Corollary~\ref{cor:single_graph}.

If $\mathit{tx}$ has support strictly less than $\tau$ in $A_r$'s snapshot, then since $\mathit{tx}' \in K_r$ the anchor SCC of $G_r$ contains some solid $Z$. By Lemma~\ref{lem:graph_directional}(2) exactly one of $(\mathit{tx}', Z)$ and $(Z, \mathit{tx}')$ is in the finalized $G_r$, so there is a path $\mathit{tx}' = u_0 \to u_1 \to \cdots \to u_l = Z$ in $G_r$. Each edge on the path corresponds to at least $n(1 - \gamma) + 1$ correct replicas agreeing on direction. Combined with Lemma~\ref{lem:solid_precedes} applied to $Z$ and $\mathit{tx}$, and with the hypothesis on $\mathit{tx}$ and $\mathit{tx}'$, the sequence $[\mathit{tx}', u_1, \dots, u_{l-1}, Z, \mathit{tx}, \mathit{tx}']$ forms a Condorcet cycle.
\end{proof}

\begin{lemma}[FairUpdate votes arrive]
\label{lem:votes_arrive}
Every parked subdag eventually accumulates votes from at least $n - f$ distinct replicas.
\end{lemma}

\begin{proof}
On parking $A_r$ with missing edges $M_r$, the fairness processor sends $\mathit{FairPropose}(r, M_r)$ to its local worker. The LOI tracker is persistent, so for each $(u, v) \in M_r$ the worker records a directed vote as soon as both endpoints are observed. Clients broadcast to all replicas and indirect entries propagate transaction observations through Narwhal's reliable broadcast, so every correct replica eventually observes both endpoints. Votes queue into the next outgoing batch and reach all correct replicas by Narwhal's validity. Within finitely many rounds, $A_r$'s parked graph accumulates votes from at least $n - f$ replicas.
\end{proof}

\subsection*{Main Theorems}

\begin{proof}[Proof of Theorem~\ref{thm:agreement_main}]
By Lemma~\ref{lem:reduction}, the emitted sequence of every correct replica matches that of the serial reference. The serial reference is a deterministic function of the committed subdag sequence and the observed FairUpdate vote set, both of which are identical across replicas by Tusk's agreement and total order properties.
\end{proof}

\begin{proof}[Proof of Theorem~\ref{thm:fairness_main}]
Let $\mathit{tx}$ finalize in $G_s$ and $\mathit{tx}'$ in $G_r$, and suppose $\gamma n$ replicas received $\mathit{tx}$ before $\mathit{tx}'$. The case analysis follows Themis Theorem~4.1~\cite{Kelkar_Cornell_Themis_Nov2022}.

If $r = s$, Lemma~\ref{lem:graph_directional}(3) rules out the wrong direction edge. Either $(\mathit{tx}, \mathit{tx}') \in G_r$, in which case $\mathit{tx}$ either topologically precedes $\mathit{tx}'$ in the condensation or shares an SCC with it and is emitted contiguously under the intra SCC linearization, or the pair is missing and resolved by ApplyFairUpdate in the correct direction. In every case $\mathit{tx}$ is emitted no later than $\mathit{tx}'$.

If $s < r$, Emit drains in subdag commit order, so $\mathit{tx}$ is output before $\mathit{tx}'$.

If $s > r$, Lemma~\ref{lem:contiguous_cycles} places $\mathit{tx}$ and $\mathit{tx}'$ in the same Condorcet cycle, so they share a batch in the maximal cyclic batch partition of the output, and $\mathit{tx}$ is emitted no later than $\mathit{tx}'$.
\end{proof}

\begin{proof}[Proof of Theorem~\ref{thm:liveness_main}]
Let $\mathit{tx}$ be submitted by a correct client. Each correct replica eventually includes $\mathit{tx}$ in a worker batch, and by Narwhal's validity the batch eventually commits. The support of $\mathit{tx}$ reaches $\tau_s$ in some snapshot, and by Lemma~\ref{lem:graph_directional}(1), $\mathit{tx} \in K_r$ for that subdag.

If $G_r$ has no missing edges it finalizes immediately. Otherwise Lemma~\ref{lem:votes_arrive} delivers $n - f$ votes, and Lemma~\ref{lem:graph_directional}(3) forces every missing pair to resolve in the correct direction. By the anchor rule, $G_r$'s finalization depends only on its own vertices, so Condorcet cycles cannot stall it. By induction on commit order every earlier subdag finalizes, and Emit eventually outputs $A_r$'s order.
\end{proof}

\section{Liveness Attacks on Fairness DAG BFT Protocols}\label{liveness_attacks}

While implementing FairDAG-RL~\cite{Kang_FairDAG_Aug2025} and
DoD~\cite{Nagda_UPenn_PhD_2025} from their pseudocode descriptions,
we identified liveness bugs in both protocols' fairness layers. In this section
we describe each bug, give a concrete attack or failure scenario, and state
the fix.

\subsection{Liveness Attack on FairDAG-RL}

In FairDAG-RL, each replica assigns monotonically increasing ordering indicators to transactions as they arrive and broadcasts them as DAG vertices. The underlying DAG consensus periodically commits a leader vertex $L_r$ whose causal history defines a subdag $A_r$ (the set of newly committed vertices not included in any earlier leader's causal history). The fairness layer receives $A_r$ and processes it as follows. First, for each transaction digest $d$ appearing in $A_r$, the layer records the ordering indicator from the proposing replica into a global vector $node(d).committed\_ois$. Second, the layer checks whether $d$ has accumulated enough ordering indicators to be promoted: a node with $ap(d, r) \geq N - f$ indicators becomes \emph{solid}, one with $ap(d, r) \geq \lceil(N{-}f)/2\rceil$ becomes \emph{shaded}, and otherwise it remains \emph{blank}, where non-blank nodes are inserted into a dependency graph $G_r$. Third, the layer updates pairwise edge weights in $G_r$ and adds a directed edge between two nodes once the number of replicas preferring one direction reaches the non-blank threshold of $\lceil(N{-}f)/2\rceil$. The fairness layer finalizes a transaction ordering once $G_r$ becomes a \emph{tournament}, i.e., every pair of nodes is connected by exactly one directed edge.

\smallskip
\noindent\textbf{Root Cause.} The weight update step (Figure~8, lines~19--32 of~\cite{Kang_FairDAG_Aug2025}) iterates over every vertex $v \in A_r$ and every transaction $d$ listed in $v$. For each such $d$, it looks up the dependency graph $G'$ to which $d$ belongs and compares $d$'s ordering indicator from the proposing replica against those of all other nodes in $G'$. Crucially, this comparison fires \emph{only} when $d$ appears in a vertex of the \emph{current} subdag $A_r$ and $d$ already belongs to some graph $G'$ at the time of processing.

This creates a gap where a certain replica $R_k$ may report transactions $d_1$ and $d_2$ in an earlier subdag $A_{r'}$, at which point both are still blank and belong to no graph, so the comparison is silently skipped. When a later subdag $A_r$ finally promotes $d_1$ and $d_2$ into a graph, $R_k$'s vertex in $A_r$ may carry unrelated transactions and thus never re-trigger the comparison. The ordering indicators from $R_k$ remain stored in $committed\_ois$ but are permanently excluded from the weight computation. If $R_k$'s vote is the decisive one needed to push either direction of the pair $(d_1, d_2)$ past the non-blank threshold, no edge is ever added between them, thus the dependency graph
never becomes a tournament, and the fairness layer stalls indefinitely.

\smallskip
\noindent\textbf{Concrete Attack.}We use $N = 4$, $f = 1$. Replica $R_4$ crashes silently, leaving three honest replicas $R_1, R_2, R_3$. The quorum threshold for adding an edge is $\lceil 3/2 \rceil = 2$ and the threshold for promoting a node to solid is $3$. A malicious client triggers the liveness attack by crafting two transactions $a$ and $b$.

\smallskip
\noindent\textbf{Subdag $A_1$.} The client submits $a$ and $b$ exclusively to $R_3$, and unrelated filler transactions to $R_1$ and $R_2$:
\begin{align*}
R_1 &: \{(x, 1)\}, \quad
R_2 : \{(x, 1)\}, \quad
R_3 : \{(a, 1),\, (b, 2)\}
\end{align*}
Only one replica ($R_3$) has seen $a$ and $b$, which is below the non-blank threshold, so both remain blank and enter no graph. However, $R_3$'s ordering ($a$ before $b$) is recorded in $committed\_ois$. The weight-update loop skips the pair because neither node belongs to a graph yet.

\smallskip
\noindent\textbf{Subdag $A_2$.} The client submits $a$ and $b$ to $R_1$ and $R_2$ in opposite orders, and keeps $R_3$ busy with a filler transaction $y$:
\begin{align*}
R_1 &: \{(a, 1),\, (b, 2)\}, \quad
R_2 : \{(b, 1),\, (a, 2)\}, \quad
R_3 : \{(y, 1)\}
\end{align*}
Now $ap(d, r=2) = 3$, i.e., three replicas have seen both $a$ and $b$, so both are promoted to solid and inserted into a new dependency graph $G_2$. The weight update loop processes $A_2$'s vertices where $R_1$ votes $a \prec b$, $R_2$ votes $b \prec a$, and now $R_3$ contributes nothing to the transaction pair (it has done so in previous sub-dag commit). The resulting weights are $G_2.weight[(a,b)] = 1$ and $G_2.weight[(b,a)] = 1$, both below the non-blank threshold of $2$, so no edge is added. The dependency graph is not a tournament and the fairness layer stalls permanently, blocking all future rounds. Had $R_3$'s earlier vote ($a \prec b$ from $A_1$) been counted, the weight $G_2.weight[(a,b)]$ would reach $2$, an edge would be added, and the graph would finalize normally.

\smallskip
\noindent\textbf{Patch.} The root cause is that the weight-update loop (Figure~8, lines~19--32
of~\cite{Kang_FairDAG_Aug2025}) only considers ordering indicators that arrive in the \emph{current} subdag, missing any that were deposited while the transaction was still blank. The fix is to add a \emph{catch-up} pass immediately after lines~11--18 (the node classification step), i.e., whenever a node $d$ is newly promoted into a graph $G_r$, compute its pairwise weights against every other node already in $G_r$ by scanning the \emph{full} $committed\_ois$ vectors of both nodes across all $N$ replicas, rather than relying on the current subdag's vertices as triggers. This is both necessary
and sufficient as every ordering indicator that was recorded while $d$ was blank is now accounted for exactly once, and the incremental loop at lines~19--32 continues to handle any future indicators arriving in later subdags. We have integrated this patch in our implementation of FairDAG-RL and verified that it restores finalization liveness.

\subsection{Liveness Attack on DoD}

In DoD, each replica independently constructs a global dependency
graph from $N{-}f$ collected local-order messages \emph{before} consensus.
For each pair of data-dependent transactions $t$ and $t'$, the replica counts
how many local orders place $t$ before $t'$ (denoted $w(t,t')$) and vice
versa. If the larger count meets or exceeds the edge threshold
$n(1{-}\gamma)+f+1$ and is greater than its counterpart, a directed edge is
added. Otherwise, the pair is recorded as a \emph{missing edge} in a set
$M$, that is then broadcast alongside the global-order graph. Each replica
locally maintains a missing edge store $M_w$ that tracks accumulated weights.
A committed global-order graph can only be executed once all its missing
edges are resolved in the executing replica's $M_w$, otherwise the execution
queue stalls (Algorithm~3, line~18 of~\cite{Nagda_UPenn_PhD_2025}).

\smallskip
\noindent\textbf{Bug~1: Weights diverge across replicas.}
Algorithm~2, line~32 of~\cite{Nagda_UPenn_PhD_2025} broadcasts missing pairs
without their weights. Since each replica constructs its global graph from a
potentially different quorum of $N{-}f$ local orders, the computed weight
for the same pair can differ. For instance, one replica may compute
$w(a,b) = 2$ and store this in its $M_w$, but broadcasts only the bare pair
$(a,b)$ without the weight. Another replica that independently computed
$w(a,b) = 1$ from a different quorum stores its own value of~$1$. Per the
protocol description, a replica increments the weight by~$1$ upon receiving
a missing pair, but even with this update applied, the two replicas' weights
remain permanently diverged.

\smallskip
\noindent\textbf{Bug~2: Weights are inflated by repeated broadcasts.}
The same mechanism also leads to double-counting. Algorithm~2, lines~39--40
of~\cite{Nagda_UPenn_PhD_2025} increment the scalar weight in $M_w$ by~$1$
for each received copy of a missing pair, with no deduplication of the
originating evidence. Consider the case where all $f$ replicas crash
silently: the only available quorum is the $N{-}f$ honest replicas, so every
honest replica collects the same set of local orders, constructs an
identical global-order graph, and broadcasts the same missing pair to all
others. Upon receiving these $N{-}f$ identical broadcasts, each replica
increments the weight $N{-}f$ times despite the underlying evidence
originating from the same set of local orders. The weight thus becomes
inflated and may spuriously cross the edge threshold in a direction that no
quorum of local orders genuinely supports.

\smallskip
\noindent\textbf{Bug~3: Weights are frozen against past-round evidence.}
The protocol provides two mechanisms to accumulate directional weight for a
missing pair $(t, t')$ in $M_w$. The first (Algorithm~1, lines~7--8
of~\cite{Nagda_UPenn_PhD_2025}) increments $w(t', t)$ when transaction $t$
physically arrives at a replica that already tracks the pair. The second
(Algorithm~2, lines~11--12 of~\cite{Nagda_UPenn_PhD_2025}) boosts the weight
when both $t$ and $t'$ appear as vertices in a future round's global
dependency graph. Both are ineffective for a common class of pairs. The
first requires $t$ to arrive \emph{after} the pair is added to $M_w$, but
missing pairs are created during global ordering, which runs after both
transactions have already been received. Since client transactions are
submitted only once, neither will arrive again. The second requires both
transactions to appear as vertices in a future round's global dependency
graph, but each transaction is included in exactly one round's local-order
graph and will not appear in any subsequent round.

\smallskip
\noindent\textbf{Concrete scenario.} We use $N = 5$, $f = 1$, $\gamma = 1$,
with $R_5$ that immediately silently crashes. The non-blank threshold is $n(1{-}\gamma) + f + 1 = 2$ and the solid threshold is $n - 2f = 3$. A client submits two
data-dependent transactions $a$ and $b$ to all four honest replicas, but
due to network asynchrony the transactions arrive in different local-order
rounds:

\begin{center}
\begin{tabular}{lcc}
\toprule
Replica & Round of $a$ & Round of $b$ \\
\midrule
$R_1$ & $r$ & $r$ \\
$R_2$ & $r{-}1$ & $r$ \\
$R_3$ & $r$ & $r{-}1$ \\
$R_4$ & $r$ & $r$ \\
\bottomrule
\end{tabular}
\end{center}

In round $r$'s global ordering, every honest replica collects the same
quorum of $N{-}f = 4$ local orders from $\{R_1, R_2, R_3, R_4\}$. Both $a$
and $b$ appear in three round-$r$ local orders ($a$ in $R_1, R_3, R_4$;
$b$ in $R_1, R_2, R_4$) and are classified as fixed. However, only $R_1$
and $R_4$ include both transactions in their round-$r$ local order; $R_2$
includes only $b$ (it placed $a$ in round $r{-}1$) and $R_3$ includes only
$a$. Suppose $R_1$ received $a$ before $b$ and $R_4$ received $b$ before
$a$. Then $w(a, b) = 1$ and $w(b, a) = 1$. Neither reaches the edge
threshold of~$2$, so $(a, b)$ is recorded as a missing edge.

Although $R_2$ and $R_3$ both have evidence of the pair from round $r{-}1$,
this information was captured in a prior round's local-order graph and is
never contributed to the pair's resolution. Neither weight accumulation
mechanism ever fires, as neither transaction will arrive again nor reappear
in a future round's dependency graph. The weights remain stuck at $1$ permanently, and the execution queue stalls on the first global-order graph containing this pair, blocking all subsequent committed graphs.

\smallskip
\noindent\textbf{Patch.} A similar implicit edge update approach of FairDAG-RL could be applied, but is non-trivial in DoD
because global graphs are constructed \emph{before} consensus. Implicitly
resolved edges may be based on local orders from up to $f$ replicas whose
vertices are ultimately not committed, as the DAG structure of DoD
does not include weak edges. This means implicitly added edges could need
to be reverted after consensus, making the fix non-intuitive. For this
reason, we have chosen to implement \emph{explicit} edge resolution, where
a replica that resolves a missing edge broadcasts the resolved direction to
all other replicas, as done in
Themis~\cite{Kelkar_Cornell_Themis_Nov2022} and Herring.

\end{document}